\documentclass[12pt]{article}
\usepackage[latin1]{inputenc}
\usepackage[british]{babel}
\usepackage{cmap}
\usepackage{lmodern}

\usepackage{amssymb, amsmath, amsthm}
\usepackage[a4paper,top=25mm,bottom=25mm,left=25mm,right=25mm]{geometry}
\usepackage{etex}
\usepackage{ragged2e}

\usepackage{authblk} 
\usepackage{pifont}
\usepackage{graphicx}
\usepackage[usenames,dvipsnames,svgnames,table]{xcolor}
\usepackage[figuresright]{rotating}
\usepackage{xtab} 
\usepackage{longtable} 
\usepackage{ltxtable} 
\usepackage{multirow}
\usepackage{footnote}
\usepackage[stable]{footmisc}
\usepackage{chngpage} 
\usepackage{pdflscape} 
\usepackage[nottoc,notlot,notlof]{tocbibind} 

\usepackage{pgfplots}
\pgfplotsset{every tick label/.append style={font=\small}}
\pgfplotsset{compat=1.14}
\usepackage{setspace}

\usepackage{array}
\newcolumntype{K}[1]{>{\centering\arraybackslash$}p{#1}<{$}}

\makesavenoteenv{tabular}
\usepackage{tabularx}
\usepackage{booktabs}
\usepackage{threeparttable}
\usepackage[referable]{threeparttablex} 
\newcolumntype{R}{>{\raggedleft\arraybackslash}X}
\newcolumntype{L}{>{\raggedright\arraybackslash}X}
\newcolumntype{C}{>{\centering\arraybackslash}X}
\newcolumntype{M}[1]{>{\centering\arraybackslash}m{#1}}
\newcolumntype{A}{>{\columncolor{gray!25}}C}
\newcolumntype{a}{>{\columncolor{gray!25}}c}

\newlength{\tablen}

\usepackage{dcolumn} 
\newcolumntype{.}{D{.}{.}{-1}}

\usepackage{tikz}
\usetikzlibrary{arrows, automata, calc, matrix, patterns, positioning, trees}
\usepackage[semicolon]{natbib}
\usepackage[hyphens]{url}
\makeatletter
\g@addto@macro{\UrlBreaks}{\UrlOrds}
\makeatother
\usepackage{hyperref} 
\hypersetup{
  colorlinks   = true,    	
  urlcolor     = blue,    	
  linkcolor    = blue,    	
  citecolor    = red		
}
\usepackage{microtype}
\usepackage[justification=centerfirst]{caption} 

\usepackage[labelformat=simple]{subcaption}

\DeclareCaptionLabelFormat{parenthesis}{(#2)}
\makeatletter
\renewcommand\p@subfigure{\arabic{figure}.}
\makeatother

\DeclareCaptionLabelFormat{parenthesis}{(#2)}
\makeatletter
\renewcommand\p@subtable{\arabic{table}.}
\makeatother

\newenvironment{customlegend}[1][]{%
	\begingroup
	\csname pgfplots@init@cleared@structures\endcsname
	\pgfplotsset{#1}%
    }{%
	\csname pgfplots@createlegend\endcsname
	\endgroup
    }%
\def\addlegendimage{\csname pgfplots@addlegendimage\endcsname}

\usepackage{enumitem}
\setlist[itemize]{leftmargin=2.5\parindent}
\setlist[enumerate]{leftmargin=2.5\parindent}

\theoremstyle{plain}

\newtheorem{proposition}{Proposition}[section]

\theoremstyle{definition}
\newtheorem{axiom}{Axiom}

\newtheorem{definition}{Definition}[section]
\newtheorem{example}{Example}[section]

\theoremstyle{remark}

\newtheorem{remark}{Remark}[section]

\def\keywords{\vspace{.5em} 
{\noindent \textit{Keywords}: }}

\def\JEL{\vspace{.5em} 
{\noindent \textbf{\emph{JEL} classification number}: }}

\def\AMS{\vspace{.5em} 
{\noindent \textbf{\emph{MSC} class}: }}

\author{\href{https://sites.google.com/site/laszlocsato87}{L\'aszl\'o Csat\'o}\thanks{~E-mail: \emph{csato.laszlo@sztaki.hu} \newline Laboratory on Engineering and Management Intelligence, Research Group of Operations Research and Decision Systems, Institute for Computer Science and Control (SZTAKI), Budapest, Hungary \newline Department of Operations Research and Actuarial Sciences, Corvinus University of Budapest (BCE), Budapest, Hungary}
$\qquad$ $\qquad$ Csaba T\'oth\thanks{~Budapest University of Technology and Economics, Budapest, Hungary \newline Institute of Economics, Research Centre for Economic and Regional Studies, Budapest, Hungary}}
\title{University rankings from the revealed \\ preferences of the applicants}
\date{\today}

\def\Dedication{
{\noindent ``The ranking never lies.''}\footnote{~Source: \url{https://www.eurosport.com/tennis/stan-wawrinka-doubts-his-world-no.1-credentials-but-is-he-right_sto5945690/story.shtml}.
Downloaded 12 September 2018.}

\flushright
\noindent (Stan Wawrinka, who won three Grand Slam tournaments in tennis)

\vspace{1cm} 
\justify }

\begin{document}

\maketitle

\Dedication

\begin{abstract}
\noindent
A methodology is presented to rank universities on the basis of the lists of programmes the students applied for. We exploit a crucial feature of the centralised assignment system to higher education in Hungary: a student is admitted to the first programme where the score limit is achieved. This makes it possible to derive a partial preference order of each applicant. Our approach integrates the information from all students participating in the system, is free of multicollinearity among the indicators, and contains few ad hoc parameters.
The procedure is implemented to rank faculties in the Hungarian higher education between 2001 and 2016. We demonstrate that the ranking given by the least squares method has favourable theoretical properties, is robust with respect to the aggregation of preferences, and performs well in practice.
The suggested ranking is worth considering as a reasonable alternative to the standard composite indices.

\keywords{Decision analysis; university ranking; group decision; preference ordering; incomplete pairwise comparison}

\AMS{62F07, 91B10, 97B40}

\JEL{C44, I23, Z18}
\end{abstract}

\clearpage
\tableofcontents
\vspace{2cm}
\listoftables
\vspace{1cm}
\listoffigures

\clearpage

\section{Introduction} \label{Sec1}

The global expansion of higher education has created an increasing demand for the comparison of universities and has inspired the development of ranking systems or league tables around the world \citep{DillSoo2005, UsherSavino2007}. These rankings are usually based on the composition of various factors, namely, they are indices with a number of moving parts, which the producer -- usually an academic institution, a government, a magazine, a newspaper, or a website -- is essentially free to set.
This approach is widely criticised for its weak theoretical link to quality and serious methodological drawbacks \citep{Brooks2005}, as well as for its sensitivity \citep{LiuCheng2005, Saisanad'HombresSaltelli2011, DeWitteHudrlikova2013, Safon2013, Moed2017, Soh2017, Johnes2018}. For example, \citet{OlcayBulu2017} find that there are significant differences even among indices focusing on the same aspect such as teaching or research output.
Therefore, similar measures have been called ``mashup indices'' in development economics \citep{Ravallion2012a}.

Nevertheless, college ranking remains a transparent tool of fair evaluation for the public, which may have a huge impact on higher education institution decision making \citep{Hazelkorn2007, Marginson2014}, so there is a clear need for more robust rankings. It seems that there exist two separate directions to obtain them.

The first line of research strives to more closely integrate university rankings with multi-criteria decision making.
According to \citet{BillautBouyssouVincke2009}, the criteria used in the Shanghai ranking are irrelevant, the whole procedure pays insufficient attention to fundamental structuring issues, while the absence of basic knowledge of aggregation techniques and their properties vitiate the evaluation of the institutions.
\citet{GiannoulisIshizaka2010} use a three-tier web-system to get a customised ranking of British Universities with ELECTRE III multi-criteria decision method that reflects personal preferences. This approach is, for instance, able to reveal any disastrous criterion with the veto threshold or distinguish between indifference and incomparability in the alternatives.
According to \citet{Kaycheng2015}, the indicators of Times Higher Education World University Ranking 2013-2014 tend to correlate with one another, causing a multicollinearity problem, furthermore, some indicators may contribute little to the overall score. The obvious solution is the identification and exclusion of redundant indicators, which are non-contributing or even misinforming.
\citet{KunschIshizaka2018} apply a model where evaluations are given by performance profile distributions to assess the quality of research in the United Kingdom since a single numerical score can hardly represent a complex criterion.
\citet{CorrenteGrecoSlowinski2019} increase the robustness of the ranking by using a multiple criteria hierarchy process and the Choquet integral preference model, which generalizes the weighted sum to take into account the possible negative and positive interactions between the criteria.

The second approach aims to handle the above problems by other methodologies, for example, through pairwise comparisons. They do not require directly an outright ranking of all universities studied, which might not be done easily, but only choices between two universities. Thus universities play virtual matches against each other, and an institution defeats another if it is preferred by a student to the other. Each applicant provides a tournament, aggregated into a common preference matrix to derive a ranking.
\citet{DittrichHatzingerKatzenbeisser1998} analyse a survey among the first-year students at the Vienna University (Austria) to show that preferences are substantially different for different groups of students.
\citet{Averyetal2013} investigate US undergraduate programs using a national sample of high-achieving students to obtain a ``desirability'' ranking of colleges, which implicitly weights all features by the degree to which the students collectively care about them.
\citet{KoczyStrobel2010} suggest a similar method for journal ranking.
Such a ranking is essentially parameter-free, independent of an arbitrary choice of factors and component weights, and may reflect all characteristics of a university that are observed by the applicants even if they are non-measurable.

In this paper, a methodology is presented to derive a higher education ranking from the lists of programmes the students applied for.
We exploit administrative data in Hungary, a country which has a centralised system of admissions to higher education designed such that it is relatively straightforward to recover the preferences. The dataset has been used recently to analyse college choices \citep{Telcsetal2015}, to identify obvious mistakes of the applicants in a strategically simple environment \citep{ShorrerSovago2019}, or to rank the universities \citep{TelcsKosztyanTorok2016}.
There exists even a whole research project devoted to the latter topic (see \url{http://ranking.elte.hu/}), led by Gy\"orgy F\'abri, who has launched the first Hungarian university ranking in 2001.

Similarly to \citet{DittrichHatzingerKatzenbeisser1998} and \citet{Averyetal2013}, our procedure uses existing methods of paired comparisons rankings.
The main contribution resides in the analysis of the dataset. While the previous literature has used survey data to obtain the preferences, our administrative dataset covers every Hungarian student applying for higher education in their home country across 16 years (2001-2016). In addition, it can be argued that the applicants have taken substantial efforts to be well-informed before such a high-stakes decision. 
On the other hand, the surveys designed by \citet{DittrichHatzingerKatzenbeisser1998} and \citet{Averyetal2013} explicitly asked the students to choose between two universities, while we should devise what a given list of admissions reveal about the preferences of the applicant.
Finally, opposed to these papers, an entire subsection is devoted to the axiomatic comparison of three paired comparisons-based ranking methods to help their understanding and the choice among them.

It seems that the suggested methodology can be applied in several other fields where applicants should reveal some preferences and the centralised allocation mechanism provides truthfulness such as the student-optimal deferred acceptance algorithm \citep{GaleShapley1962, DubinsFreedman1981}. Besides Hungary (see the detailed discussion in Section~\ref{Sec4}), college admission is organised basically along these principles in Chile \citep{RiosLarroucauParraCominetti2014}, Ireland \citep{Chen2012}, and Spain \citep{Romero-Medina1998}.
School choice often shows similar characteristics, too, like in New York, Hungary, Finland, Amsterdam \citep{Biro2017}. Allocation of graduates according to their preferences in a centrally coordinated way is typical in certain professions, including residents (junior doctors) in the UK and the US, teachers in France, or lawyers in Germany \citep{Biro2017}. \citet{Biro2017} gives a comprehensive review of matching models under preferences.

Naturally, the question arises what is the advantage of such a ranking based on the revealed preferences of all applicants. Indeed, probably each student has somewhat different preferences, so the importance of criteria to be taken into account should be customised. However, we think the abundance of world university rankings clearly illustrates that there is a strong demand for unique rankings by the decision-makers and media, even if most experts are convinced that constructing multiple rankings would be more professional. In this case, the suggested ranking is worth considering as a reasonable alternative to the standard composite indices.

The paper proceeds as follows.
Section~\ref{Sec2} describes the data. The methods used to derive a ranking from the preferences are detailed in Section~\ref{Sec3}.
Our main theoretical innovation, that is, recovering the preferences from the data is presented in Section~\ref{Sec4}.
Section~\ref{Sec5} presents the results, while Section~\ref{Sec6} offers some concluding thoughts.

\section{Data} \label{Sec2}

In Hungary, the admission procedure of higher education institutions has been organised by a centralised matching scheme at a national level since 1996 \citep{Biro2008, Biro2011}.
At the beginning of the procedure, the students submit their ranking lists over the fields of studies of particular faculties they are applying for. A single application in a given year consists of the following data:
\begin{itemize}
\item
Student ID;
\item
Position according to the student's preference order;
\item
Faculty/School of a higher education institution;
\item
Course;
\item
Level of study (BSc/BA, MSc/MA, Single long cycle);
\item
Form of study (full-time training, correspondence training, evening training);
\item
Financing of the tuition (state-financed: completely financed by the state, student-financed: partly financed by the student).
\end{itemize}
For example, the record
\[
\left[
\begin{array}{ccccccc}  
158 & 2 & \text{PTE--AOK} & \text{Medicine} & O & N & A \\
\end{array}
\right]
\]
means that the student with the ID \#158 applied at the second place ($2$) for the faculty PTE--AOK, course Medicine in level $O$ (Single long cycle giving an MSc degree), form $N$ (full time), financed by the state ($A$).
In the following, the faculty, course, level of study, form of study, and financing of the tuition (e.g. PTE--AOK, Medicine, $O$, $N$, $A$) will be called a \emph{programme}.

After that, the students receive scores at each of the programmes they applied for, based on their final grades at secondary school and entrance exams. The scores of an applicant can differ in two programmes, for example, when the programmes consider the grades from different sets of subjects.

Finally, a government organisation collects all ranking lists and scores, and a centralised algorithm determines the score limits of the programmes. A student is admitted to the first programme on her list where the score limit is achieved, meaning that it is not allowed to overwrite this matching by any participant.
The score limits are good indicators of the quality and popularity of the programmes, they highly correlate with the applicants' preferences and with the job market perspectives of the graduates \citep{Biro2011}.

Our dataset contains almost all applications between 2001 and 2014: for applicants who have a list with more than six items, only the first six programmes plus the one where (s)he is admitted (if there exists such a programme) were recorded. The 2015 and 2016 data do not have this limitation.

\section{Methodology} \label{Sec3}

First, the theoretical background of ranking from pairwise comparisons will be presented.

\subsection{Ranking problems and scoring methods} \label{Sec31}

Let $N = \{ X_1,X_2, \dots ,X_n \}$ be the \emph{set of objects} on which the preferences of the agents are being expressed, and $A = \left[ a_{ij} \right] \in \mathbb{R}^{n \times n}$ be the \emph{preference matrix} such that $a_{ij} \in \mathbb{R}$ is a measure of how object $X_i$ is preferred over $X_j$. $a_{ii} = 0$ is assumed for all $X_i \in N$.

The pair $(N,A)$ is called a \emph{ranking problem}. The set of ranking problems with $n$ objects ($|N| = n$) is denoted by $\mathcal{R}^n$.

There exists a one-to-one correspondence between preference matrices and directed weighted graphs without loops: if object $X_i$ is preferred to object $X_j$ with an intensity of $a_{ij}$, then graph $G$ contains a directed edge from node $X_i$ to node $X_j$ which has a weight of $a_{ij}$, and vice versa.

The aim is to derive a \emph{ranking} of the objects from any ranking problem $(N,A)$, for which purpose scoring methods will be used.
A scoring method $f:\mathcal{R}^n \to \mathbb{R}^n$ is a function that associates a score $f_i(N,A)$ for each object $X_i$ in any ranking problem $(N,A) \in \mathcal{R}^n$. It immediately induces a ranking of the objects in $N$ (a transitive and complete weak order on the set of $N \times N$) by $f_i(N,A) \geq f_j(N,A) \Rightarrow X_i \succeq X_j$, that is, object $X_i$ is at least as good as $X_j$ if its score is not smaller.

A ranking problem $(N,A)$ has the skew-symmetric \emph{results matrix} $R = A - A^\top \in \mathbb{R}^{n \times n}$, and the symmetric \emph{matches matrix} $M = A + A^\top \in \mathbb{R}^{n \times n}$, where $m_{ij}$ is the number of comparisons between $X_i$ and $X_j$ whose outcome is given
by $r_{ij}$.

Let $\mathbf{e} \in \mathbb{R}^n$ denote the column vector with $e_i = 1$ for all $i = 1,2, \dots ,n$.

Perhaps the most straightforward measure for the goodness of the objects is the sum of their ``net'' preferences.

\begin{definition} \label{Def31}
\emph{Row sum}:
Let $(N,A) \in \mathcal{R}^n$ be a ranking problem.
The \emph{row sum} score $s_i(A)$ of object $X_i \in N$ is given by $\mathbf{s}(A) = R \mathbf{e}$, that is, $s_i(A) = \sum_{j=1}^n r_{ij}$ for all $X_i \in N$.
\end{definition}

It is clear that the row sum score does not take into account the ``popularity'' of the objects which can be a problem as the volatility of row sum for objects with a high number of comparisons is usually significantly higher than the volatility of row sum for objects with a low number of comparisons.

This effect can be handled by normalisation. Denote the degree of node $X_i$ in graph $G$ by $d_i = \sum_{X_j \in N} m_{ij}$. Introduce the diagonal matrix $D^{-} = \left[ d^{-}_{ij} \right] \in \mathbb{R}^{n \times n}$ such that $d^{-}_{ii} = 1 / d_{i}$ for all $i = 1,2, \dots ,n$ and $d^{-}_{ij} = 0$ if $i \neq j$.

\begin{definition} \label{Def32}
\emph{Normalised row sum}:
Let $(N,A) \in \mathcal{R}^n$ be a ranking problem.
The \emph{normalised row sum} score $p_i(A)$ of object $X_i \in N$ is given by $\mathbf{p}(A) = D^{-} \mathbf{s} = D^{-} R \mathbf{e}$, that is, $p_i(A) = \sum_{j=1}^n r_{ij} / d_i$ for all $X_i \in N$.
\end{definition}

In addition, a preference over a ``strong'' object is not necessarily equal to a preference over a ``weak'' object. It can be taken into account by considering the entire structure of the comparisons.

The Laplacian matrix $L = \left[ \ell_{ij} \right] \in \mathbb{R}^{n \times n}$ of the graph $G$ is given by $\ell_{ij} = -m_{ij}$ for all $i \neq j$ and $\ell_{ii} = d_{i}$ for all $i = 1,2, \dots n$.

\begin{definition} \label{Def33}
\emph{Least squares}:
Let $(N,A) \in \mathcal{R}^n$ be a ranking problem. The least squares scores can be obtained via the following least squared errors estimation:
\[
\min_{\mathbf{q} \in \mathbb{R}^n} \sum_{X_i, X_j \in N} m_{ij} \left[ \frac{r_{ij}}{m_{ij}} - \left( q_i - q_j \right) \right]^2.
\]
The \emph{least squares} score $q_i(A)$ of object $X_i \in N$ is given by $L \mathbf{q}(A) = \mathbf{s}(A) = R \mathbf{e}$, that is, $d_i q_i(A) = s_i(A) + \sum_{j=1}^n m_{ij} q_{j}(A)$ for all $X_i \in N$.
\end{definition}

The Laplacian matrix $L$ is a singular matrix, its rank equals $n-k$, where $k$ is the number of (weakly) connected components in the graph $G$. Consequently, the system of linear equations in Definition~\ref{Def33} does not have a unique solution. This can be ensured by adding the equation $\sum_{X_i \in K} q_i = 0$ for each connected component of nodes $K \subseteq N$ \citep{KaiserSerlin1978, ChebotarevShamis1999, BozokiFulopRonyai2010, CaklovicKurdija2017}.

An extensive analysis and a graph interpretation of the least squares method, as well as an overview of its origin, is provided in \citet{Csato2015a}.

The procedure is also known as the Potential Method \citep{CaklovicKurdija2017}, or as the Logarithmic Least Squares Method on the field of pairwise comparison matrices \citep{BozokiFulopRonyai2010}.
\citet{LundySirajGreco2017} prove that the least squares scores are equal to the preference vector calculated from the spanning trees of the network, while \citet{BozokiTsyganok2019} extend this result to incomplete data.
\citet{Csato2018c} and \citet{Csato2019a} provide axiomatic characterizations in the case of complete preference lists.

\subsection{An axiomatic comparison} \label{Sec32}

In the following, some axiomatic properties are presented for the ranking of the objects to illustrate the differences between the three scoring methods.

\begin{axiom} \label{Axiom1}
\emph{Size invariance}:
A scoring method $f: \mathcal{R}^n \to \mathbb{R}^n$ is said to be \emph{size invariant} if $f_i(N,A) = f_j(N,A)$ holds for any ranking problem $(N,A) \in \mathcal{R}^n$ which has two different objects $X_i,X_j \in N$ such that $a_{ik} = \alpha a_{jk}$ and $a_{ki} = \alpha a_{kj}$ for all $X_k \in N \setminus \{ X_i, X_j \}$, furthermore, $a_{ij} = a_{ji} = 0$.
\end{axiom}

Size invariance means that if there exist two objects $X_i, X_j$ with exactly the same preference structure against any third object $X_k$, but one of them is $\alpha$ times ``larger'', then they should have the same rank.

\begin{proposition} \label{Prop31}
The row sum method violates size invariance. \\
The normalised row sum and least squares methods satisfy size invariance.
\end{proposition}

\begin{proof}
\emph{Row sum}:
It can be checked that $s_i = \alpha s_j$ for two different objects $X_i,X_j \in N$ such that $a_{ik} = \alpha a_{jk}$ and $a_{ki} = \alpha a_{kj}$ for all $X_k \in N \setminus \{ X_i, X_j \}$, furthermore, $a_{ij} = a_{ji} = 0$.

\emph{Normalised row sum}:
$d_i = \alpha d_j$ also holds for two different objects $X_i,X_j \in N$ such that $a_{ik} = \alpha a_{jk}$ and $a_{ki} = \alpha a_{kj}$ for all $X_k \in N \setminus \{ X_i, X_j \}$, furthermore, $a_{ij} = a_{ji} = 0$.

\emph{Least squares}:
It follows from $m_{ik} = \alpha m_{jk}$ for all $X_k \in N$, which holds if there are two different objects $X_i,X_j \in N$ such that $a_{ik} = \alpha a_{jk}$ and $a_{ki} = \alpha a_{kj}$ for all $X_k \in N \setminus \{ X_i, X_j \}$, furthermore, $a_{ij} = a_{ji} = 0$.
\end{proof}

Size invariance is a desirable property in the context of university rankings because a method violating this property may favour institutions only because of their size, as we will see later in the case of row sum.

\begin{definition} \label{Def34}
\emph{Bridge set}:
Let $(N,A) \in \mathcal{R}^n$ be a ranking problem. An object set $\emptyset \neq B \subseteq N$ is called \emph{bridge set} if there exists $N^1, N^2 \subseteq N$ such that $N^1 \cup B \cup N^2 = N$, $N^1 \cap N^2 = \emptyset$, and $m_{ij} = 0$ for all $X_i \in N^1$ and $X_j \in N^2$, furthermore, $m_{ik} = m_{i \ell}$ for all $X_i \in N^1$ and $X_k, X_\ell \in B$.
\end{definition}

\begin{remark}
The concept of bridge set is a common generalisation of \emph{bridge player} (when $|B|=1$) \citep{Gonzalez-DiazHendrickxLohmann2014} and \emph{macrovertex} (when $N^2 = \emptyset$) \citep{Chebotarev1994, Csato2019d}.
\end{remark}

\begin{axiom} \label{Axiom2}
\emph{Bridge set independence}:
A scoring method $f: \mathcal{R}^n \to \mathbb{R}^n$ is said to be \emph{bridge set independent} if $f_i(N,A) \geq f_j(N,A) \iff f_i(N,A') \geq f_j(N,A')$ holds for all $X_i, X_j \in N^1$ in the case of any two ranking problems $(N,A),(N,A') \in \mathcal{R}^n$ with a bridge set $B$ such that $a_{k \ell} = a_{k \ell}'$ for all $\{ X_k,X_\ell \} \cap N^1 \neq \emptyset$.
\end{axiom}

Bridge set independence means that the order of objects in the set $N^1$  is independent of the preferences between the objects outside $N^1$.

\begin{remark}
\emph{Macrovertex independence} \citep{Chebotarev1994, Csato2019d} is a particular case of bridge set independence when $N^2 = \emptyset$.
\end{remark}

\begin{axiom} \label{Axiom3}
\emph{Bridge set autonomy}:
A scoring method $f: \mathcal{R}^n \to \mathbb{R}^n$ is said to be \emph{bridge set autonomous} if $f_k(N,A) \geq f_\ell(N,A) \iff f_k(N,A') \geq f_\ell(N,A')$ holds for all $X_k, X_\ell \in B \cup N^2$ in the case of any two ranking problems $(N,A),(N,A') \in \mathcal{R}^n$ with a bridge set $B$ such that $a_{ij} = a_{ij}'$ for all $\{ X_i,X_j \} \cap \left( B \cup N^2 \right) \neq \emptyset$.
\end{axiom}

Bridge set autonomy requires the order of objects to remain the same in the set $B \cup N^2$ if only the preferences inside $N^1$ change.

\begin{remark}
Bridge set autonomy is an extension of \emph{macrovertex autonomy} \citep{Csato2019d}, which requires that $N^2 = \emptyset$.
\end{remark}

\begin{proposition} \label{Prop32}
The row sum and normalised row sum methods satisfy bridge set independence and bridge set autonomy.
\end{proposition}

\begin{proof}
\emph{Bridge set independence}:
It can be seen that $s_i(A) = s_i(A')$ and $d_i(A) = d_i(A')$ hold for all $X_i \in N^1$ in the case of any two ranking problems $(N,A),(N,A') \in \mathcal{R}^n$ with a bridge set $B$ such that $a_{k \ell} = a_{k \ell}'$ for all $\{ X_k,X_\ell \} \cap N^1 \neq \emptyset$.

\emph{Bridge set autonomy}:
It can be checked that $s_k(A) = s_k(A')$ and $d_k = d_k$ hold for all $X_k \in \left( B \cup N^2 \right)$ in the case of any two ranking problems $(N,A), (N,A') \in \mathcal{R}^n$ with a bridge set $B$ such that $a_{ij} = a_{ij}'$ for all $\{ X_i,X_j \} \cap \left( B \cup N^2 \right) \neq \emptyset$.
\end{proof}

\begin{proposition} \label{Prop33}
The least squares method satisfies bridge set independence and bridge set autonomy.
\end{proposition}

\begin{proof}
\emph{Bridge set independence}:
Consider the linear equations for an arbitrary object $X_k \in B$ and $X_g \in N^2$, respectively:
\begin{equation} \label{eq1}
d_k q_k - \sum_{X_i \in N^1} m_{ki} q_i - \sum_{X_\ell \in B} m_{k \ell} q_\ell - \sum_{X_h \in N^2} m_{kh} q_h = s_k;
\end{equation}
\begin{equation} \label{eq2}
d_g q_g - \sum_{X_\ell \in B} m_{g \ell} q_\ell - \sum_{X_h \in N^2} m_{gh} q_h = s_g.
\end{equation}

Note that $m_{ki} = m_{\ell i} = \bar{m}_i$ for all $X_i \in N^1$ and $X_k, X_\ell \in B$ since $B$ is a bridge set.
Sum up the $|B|$ equations of type~\eqref{eq1} and the $|N^2|$ equations of type~\eqref{eq2}, which leads to:
\begin{equation} \label{eq3}
\sum_{X_i \in N^1} \bar{m}_i \left( \sum_{X_k \in B} q_k - |B| q_i \right) = - \sum_{X_i \in N^1} s_i.
\end{equation}
Take also the linear equation for an arbitrary object $X_i \in N^1$:
\begin{equation} \label{eq4}
d_i q_i - \sum_{X_j \in N^1} m_{ij} q_j - \sum_{X_k \in B} \bar{m}_i q_k = s_i.
\end{equation}

$\sum_{X_k \in B} q_k$ can be substituted from equation~\eqref{eq3} into the $|N^1|$ equations of type~\eqref{eq4}. This system consists of $|N^1|$ equations and the same number of unknowns. It should have a unique solution since it has been obtained by pure substitution of formulas. As the coefficients of the system do not depend on the preferences outside the object set $N^1$, the weights of these objects are the same in the ranking problems $(N,A)$ and $(N,A')$.

\emph{Bridge set autonomy}: Analogously to the proof of bridge set independence, $\sum_{X_k \in B} q_k$ can be substituted from equation~\eqref{eq3} into the $|N^1|$ equations of type~\eqref{eq4}. This system consists of $|N^1|$ equations and the same number of unknowns. It should have a unique solution, from which $q_i(A)$ and $q_i(A')$ can be obtained for all $X_i \in N^1$, respectively. 

Introduce the notation $\Delta q_i = q_i(A') - q_i(A)$ for all $X_i \in N$. $\sum_{X_i \in N^1} s_i(A) = \sum_{X_i \in N^1} s_i(A')$ holds because only the preferences inside the object set $N^1$ may change, so equation~\eqref{eq3} implies that
\begin{equation} \label{eq5}
\sum_{X_k \in B} \Delta q_k = \frac{\sum_{X_i \in N^1} \bar{m}_i \Delta q_i}{\sum_{X_i \in N^1} \bar{m}_i} = \beta.
\end{equation}

It will be shown that $\Delta q_k = \Delta q_g = \beta / |B|$ for all $X_k \in B$ and $X_g \in N^2$. Since the system of linear equations $L \mathbf{q}(A) = \mathbf{s}(A)$ has a unique solution after normalisation, it is enough to prove that $\Delta q_k = \Delta q_g = \beta / |B|$ satisfies equations of types~\eqref{eq1} and \eqref{eq2}. The latter statement comes from $d_g = \sum_{X_\ell \in B} m_{g \ell} + \sum_{X_h \in N^2} m_{gh}$ as there are no preferences between the objects in sets $N^1$ and $N^2$.

Take an equation of type~\eqref{eq1} and note that $s_k(A') - s_k(A) = 0$ because only the preferences inside the object set $N^1$ may change:
\begin{equation} \label{eq6}
d_k \Delta q_k - \sum_{X_i \in N^1} \bar{m}_i \Delta q_i - \sum_{X_\ell \in B} m_{k \ell} \Delta q_\ell - \sum_{X_h \in N^2} m_{kh} \Delta q_h = 0.
\end{equation}
Since $d_k = \sum_{X_i \in N^1} \bar{m}_i + \sum_{X_\ell \in B} m_{k \ell} + \sum_{X_h \in N^2} m_{kh}$, with the use of the assumption $\Delta q_k = \Delta q_g = \beta / |B|$ for all $X_k \in B$ and $X_g \in N^2$, we get
\begin{equation} \label{eq7}
\sum_{X_i \in N^1} \bar{m}_i \frac{\beta}{|B|} - \sum_{X_i \in N^1} \bar{m}_i \Delta q_i = 0,
\end{equation}
which holds due to the definition of $\beta$ in equation~\eqref{eq5}. This completes the proof.
\end{proof}

\begin{axiom} \label{Axiom4}
\emph{Bridge player independence} \citep{Gonzalez-DiazHendrickxLohmann2014}:
A scoring method $f: \mathcal{R}^n \to \mathbb{R}^n$ is said to be \emph{bridge player independent} if $f_i(N,A) \geq f_j(N,A) \iff f_i(N,A') \geq f_j(N,A')$ holds for all $X_i, X_j \in \left( N^1 \cup B \right)$ in the case of any two ranking problems $(N,A),(N,A') \in \mathcal{R}^n$ with a bridge set $|B| = 1$ such that $a_{k \ell} = a_{k \ell}'$ for all $\{ X_k,X_\ell \} \cap N^1 \neq \emptyset$.
\end{axiom}

According to bridge player independence, in a hypothetical world consisting of two sets of objects connected only by a single object called bridge player, the relative rankings within each set of objects are not influenced by the preferences among the objects in the other set.

\begin{proposition} \label{Prop34}
The least squares method satisfies bridge player independence.
\end{proposition}

\begin{proof}
See \citet[Proposition~6.1]{Gonzalez-DiazHendrickxLohmann2014}.
\end{proof}

Similarly to macrovertex independence and macrovertex autonomy, we have attempted to generalise bridge player independence in a way that the extended property is satisfied by the least squares method without success.

Bridge set independence, bridge set autonomy, and bridge player independence are relevant properties for university rankings, although their conditions seldom hold in practice. Intuitively, these axioms suggest that if two sets of institutions are considered, then the relative rankings within each set are not much influenced by the preferences inside the other set (bridge set independence and autonomy). Furthermore, after fixing the preferences within both sets, the positions of universities from either set in the overall ranking are mainly determined by the preferences between the two sets. For example, the relative ranking among the faculties of engineering should be almost independent of the preferences involving other faculties, and their positions in the overall ranking should not be influenced by the preferences among the faculties of other fields when all the three properties are satisfied.

Now an illustration is provided for the three scoring methods and the four axioms.

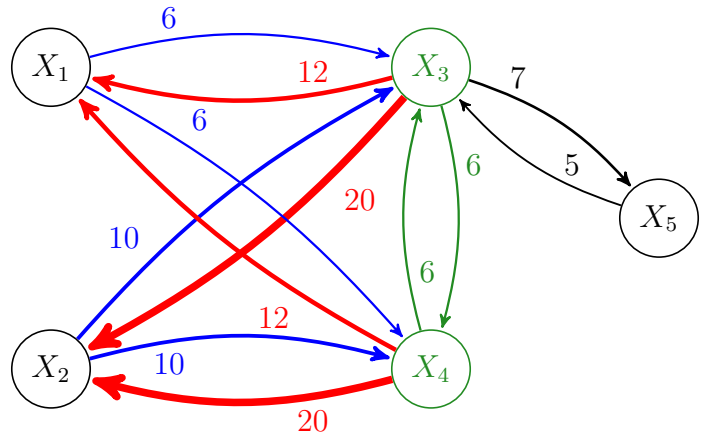
\begin{figure}[!htbp]
\centering
\begin{tikzpicture}[scale=1, auto=center, transform shape, >=triangle 45, ->, >=stealth', semithick, shorten >=1pt]
\tikzstyle{every state}=[draw,shape=circle];
  \node[state] (n1) at (0,4) {$X_1$};
  \node[state] (n2) at (0,0) {$X_2$};
  \node[state, ForestGreen] (n3) at (5,4) {$X_3$};
  \node[state, ForestGreen] (n4) at (5,0) {$X_4$};
  \node[state] (n5) at (8,2) {$X_5$};

  \path
(n1) edge [bend left=15, line width=0.3mm, blue] node [above, near start] {6}  (n3)
(n3) edge [bend left=15, line width=0.6mm, red] node [above, near start] {12} (n1)
(n2) edge [bend left=10, line width=0.5mm, blue] node [above left, near start] {10} (n3)
(n3) edge [bend left=10, line width=1.0mm, red] node [below right, near start] {20} (n2)
(n1) edge [bend left=10, line width=0.3mm, blue] node [above right, near start] {6}  (n4)
(n4) edge [bend left=10, line width=0.6mm, red] node [below left, near start] {12} (n1)
(n2) edge [bend left=15, line width=0.5mm, blue] node [below, near start] {10} (n4)
(n4) edge [bend left=15, line width=1.0mm, red] node [below, near start] {20} (n2)
(n3) edge [bend left=15, line width=0.3mm, ForestGreen] node [right, near start] {6} (n4)
(n4) edge [bend left=15, line width=0.3mm, ForestGreen] node [right, near start] {6}  (n3)
(n3) edge [bend left=15, line width=0.35mm] node [above, near start] {7} (n5)
(n5) edge [bend left=15, line width=0.25mm] node [above, near start] {5}  (n3);
\end{tikzpicture}
\caption{Preferences between the objects in Example~\ref{Examp31}}
\label{Fig1}
\end{figure}

\begin{example} \label{Examp31}
Consider the ranking problem $(N,A) \in \mathcal{R}^5$ shown in Figure~\ref{Fig1}, where the directed edges represent the preferences with their weights written on near the start of the arrow, and their thickness is proportional to the strength of the preferences.

The preference, the results, and the matches matrices are as follows:
\[
A = \left[
\begin{array}{K{2em} K{2em} K{2em} K{2em} K{2em}}  
    0     & 0     & 6     & 6     & 0 \\
    0     & 0     & 10    & 10    & 0 \\
    12    & 20    & 0     & 6     & 7 \\
    12    & 20    & 6     & 0     & 0 \\
    0     & 0     & 5     & 0     & 0 \\
\end{array}
\right]
\text{, }
R = \left[
\begin{array}{K{2em} K{2em} K{2em} K{2em} K{2em}} 
    0     & 0     & -6    & -6    & 0 \\
    0     & 0     & -10   & -10   & 0 \\
    6     & 10    & 0     & 0     & 2 \\
    6     & 10    & 0     & 0     & 0 \\
    0     & 0     & -2    & 0     & 0 \\
\end{array}
\right],
\text{ and} 
\]
\[
M = \left[
\begin{array}{K{2em} K{2em} K{2em} K{2em} K{2em}} 
    0     & 0     & 18    & 18    & 0 \\
    0     & 0     & 30    & 30    & 0 \\
    18    & 30    & 0     & 12    & 12 \\
    18    & 30    & 12    & 0     & 0 \\
    0     & 0     & 12    & 0     & 0 \\
\end{array}
\right].
\]

Objects $X_1$ and $X_2$ satisfy the conditions of the axiom size invariance, $X_2$ is $\alpha = 5/3$ times larger than object $X_1$. It remains true even if the preferences between the other three objects (the green and black directed edges) change.
For the sake of visibility, the favourable preferences of objects $X_1$ and $X_2$ are indicated by blue, and their unfavourable preferences by red colour.

Objects $X_3$ and $X_4$ form a bridge set $B$ with $N^1 = \{ X_1, X_2 \}$ and $N^2 = \{ X_5 \}$.
It holds even if the preferences in the object set $B \cup N^2 = \{ X_3, X_4, X_5 \}$ (the green and black directed edges) change. In addition, only the sum of total preferences concerning objects $X_1$ and $X_2$ should remain the same, that is, weights can be redistributed on the blue and red edges between the same nodes.

Object $X_3$ is a bridge player with $N^1 = \{ X_1, X_2, X_4 \}$ and $N^2 = \{ X_5 \}$.

\begin{table}[ht!]
\caption{Scores of the objects in Example~\ref{Examp31}}
\label{Table1}
\centering
\begin{tabularx}{0.8\textwidth}{c RRR} \toprule
Object & Row sum & Norm. row sum & Least squares  \\ \midrule
$X_1$ & $-12$  & $-20/60$    & $-1/6$  \\ 
$X_2$ & $-20$  & $-20/60$    & $-1/6$  \\ 
$X_3$ & $18$   & $15/60$	 & $1/6$   \\  
$X_4$ & $16$   & $16/60$     & $1/6$   \\ 
$X_5$ & $-2$   & $-10/60$    & $0$     \\ \bottomrule
\end{tabularx} 
\end{table}

The scores according to the three methods are shown in Table~\ref{Table1}. The row sum and least squares scores sum to $0$, but this condition does not hold for the normalised row sum.

Size invariance implies that objects $X_1$ and $X_2$ have the same rank. According to Proposition~\ref{Prop31}$  $, this is satisfied only by the normalised row sum and least squares methods. Therefore, it is difficult to argue for the row sum.

According to bridge set independence, the indifference $X_1 \sim X_2$ should be preserved even if the preferences between the other three objects (the green and black directed edges) change.

Compare the normalised row sum and least squares scores. The first favours object $X_4$ over $X_3$, but the sole difference between them is that $X_3$ has some extra preferences with object $X_5$, which seems to be a weak reason to distinguish between the strength of $X_3$ and $X_4$. This is provided by the axiom bridge player independence.

Finally, bridge player autonomy ensures that the relative ranking of $X_3$, $X_4$, and $X_5$ is not influenced by the existence of direct preferences between the objects $X_1$ and $X_2$.
\end{example}

\begin{proposition} \label{Prop35}
The row sum and normalised row sum methods violate bridge player independence.
\end{proposition}

\begin{proof}
Consider the ranking problem $(N,A) \in \mathcal{R}^5$ of Example~\ref{Examp31}.

Let $(N,A')$ be the ranking problem such that the preference matrix $A'$ is the same as $A$ except for $a'_{35} = 3 < 7 = a_{35}$. Then $s_3(A') = 14 < 16 = s_4(A')$, while $s_3(A) = 18 > 16 = s_4(A)$, showing the violation of bridge player independence by the row sum method.

Let $(N,A'')$ be the ranking problem such that the preference matrix $A''$ is the same as $A$ except for $a''_{35} = 10 > 7 = a_{35}$. Then $p_3(A'') = 7/25 > 16/60 = p_4(A'')$, while $p_3(A) = 15/60 < 16/60 = p_4(A)$, showing the violation of bridge player independence by the normalised row sum method.
\end{proof}

\begin{table}[ht!]
\caption{Axiomatic comparison of scoring methods}
\label{Table2}
\centering
\begin{tabularx}{1\textwidth}{lCcC} \toprule
    Axiom & Row sum & Normalised row sum & Least squares \\ \midrule
    Size invariance & \textcolor{BrickRed}{\ding{55}} & \textcolor{ForestGreen}{\ding{52}} & \textcolor{ForestGreen}{\ding{52}} \\
    Bridge set independence & \textcolor{ForestGreen}{\ding{52}} & \textcolor{ForestGreen}{\ding{52}} & \textcolor{ForestGreen}{\ding{52}} \\
    Bridge set autonomy & \textcolor{ForestGreen}{\ding{52}} & \textcolor{ForestGreen}{\ding{52}} & \textcolor{ForestGreen}{\ding{52}} \\
    Bridge player independence & \textcolor{BrickRed}{\ding{55}} & \textcolor{BrickRed}{\ding{55}} & \textcolor{ForestGreen}{\ding{52}} \\ \bottomrule
\end{tabularx} 
\end{table}

Table~\ref{Table2} summarises the findings of our concise axiomatic analysis. It turns out that the least squares method shares the advantages of the other two procedures, while it is the only one satisfying bridge player independence. \citet{Gonzalez-DiazHendrickxLohmann2014} and \citet{CaklovicKurdija2017} discuss some further properties of this method, mainly with positive conclusions.

Hence, at this stage from theoretical reasons, we suggest applying the least squares method for ranking in similar problems.
However, any other well-known scoring procedures can be applied to rank the objects \citep{ChebotarevShamis1998a, ChebotarevShamis1999, Palacios-HuertaVolij2004, SlutzkiVolij2005, Gonzalez-DiazHendrickxLohmann2014, Kitti2016, BubboloniGori2018}, keeping in mind that there does not exist a perfect solution \citep{Csato2019d}.

The least squares method has also a growing list of successful applications, including ranking historical Go \citep{ChaoKouLiPeng2018} and tennis players \citep{BozokiCsatoTemesi2016}, teams in Swiss-system chess tournaments \citep{Csato2017c}, or the participating countries of the Eurovision Song Contest \citep{CaklovicKurdija2017}. It is used in international price comparisons by the OECD \citep{EltetoKoves1964,Szulc1964}, to evaluate movies on a subset of Netflix data \citep{JiangLimYaoYe2011}, and for obtaining an alternative quality of life ranking \citep{Petroczy2018, Petroczy2019}.

\section{Recovering preferences from the list of applications} \label{Sec4}

Our central assumption is that the applications of a student partially reveal real preferences.
This is far from true in the case of school choice mechanisms in general \citep{AbdulkadirougluSonmez2003}. However, the Hungarian centralised matching scheme applies the Gale-Shapley algorithm at its core \citep{BiroFleinerIrvingManlove2010, BiroKiselgof2015, AgostonBiroMcBride2016}, its college-oriented version until 2007, and the applicant-oriented variant since then \citep{Biro2008}.
In the applicant-oriented Gale-Shapley algorithm \citep{GaleShapley1962}, students cannot improve their fate by lying about their preferences \citep{DubinsFreedman1981}. While the college-oriented version does not satisfy this property, the difference of the two versions is negligible in practice, and a successful manipulation requires a lot of information, which is nearly impossible to obtain \citep{TeoSethuramanTan2001}.\footnote{~The actual algorithm is a heuristic that is close to the Gale-Shapley algorithm because of the existence of some specialities like ties or common quotas \citep{BiroFleinerIrvingManlove2010, BiroKiselgof2015, AgostonBiroMcBride2016}. Consequently, it remains strategy-proof only essentially, for example, when the applicants believe they have no chance to influence the score limits.}

Naturally, the whole preference list of an applicant always remains hidden. The exact rules governing the length of the rankings changed several times between 2001 and 2016. In the first years, there was no limit on the number of applications from a given student, but they were charged for each item after the first three. In recent years, it has been allowed to apply for at most five (six in 2016) places for a fixed price such that the state-financed and student-financed versions of the same programme count as one.

As a consequence, the applications of a student reveal only a part of her preferences. In the presence of such constrained lists, \citet[Proposition~4.2]{HaeringerKlijn2009} show that -- when the centralised allocation rule is the student-optimal Gale-Shapley algorithm -- the applicant can do no better than selecting some programmes among the acceptable ones and ranking them according to the real preferences.
Thus, it is assumed for a student that:
\begin{enumerate}
\item
(S)he \emph{prefers} an object to any other object having a worse position on her list of applications;
\item
Her preference between an object on her list of applications and an object not on her list of applications is \emph{unknown};
\item
Her preference between two objects not on her list of applications is \emph{unknown}.
\end{enumerate}

\citet{TelcsKosztyanTorok2016} do not follow our second and third assumptions, therefore, according to \citet[Proposition~4.2]{HaeringerKlijn2009}, the individual choices derived by them are not guaranteed to reflect the real preferences of the applicants, in contrast to our model. For example, unranked objects cannot be legitimately said to be less preferred than any of the ones on the list, since a student may disregard a programme when (s)he knows that (s)he has no chance to be admitted there.

Due to the rules of the centralised system, the same object may appear more than once on the list of a student. For instance, if courses are compared, then both the state-financed and the student-financed versions of a particular course may be present. Then we preserve only the first appearance of the given object and delete all of the others. This ensures that each applicant can have at most one preference between two different objects.
Furthermore, only the preferences concerning the first appearance of a given object reflect the real preferences of the applicant adequately. Assume that (s)he prefers faculty $A$ to $B$, so her list contains a particular programme of faculty $A$ in the first place and the same programme of faculty $B$ in the second place. Faculty $A$ also offers another programme, which is not the favourite of the student but (s)he applies for it because, for example, (s)he can achieve the score limit with a higher probability. Then it cannot be said that faculty $B$ is preferred to faculty $A$ in any sense.

However, the financing of the tuition may somewhat distort the picture.
Suppose that an applicant prefers faculty $A$ to faculty $B$. Both faculties offer the same programme in state-financed and student-financed forms such that the score limit of the former is above the score limit of the latter as natural. The applicant knows that (s)he has no chance to be admitted to the state-financed programme of faculty $A$, but (s)he wants to avoid paying the tuition, therefore her list contains the state-financed programme of faculty $B$ in the first place, the student-financed programme of faculty $A$ in the second, and the student-financed programme of faculty $B$ in the third position. Then our technique concludes that faculty $B$ is preferred to faculty $A$, which is opposite to the real preferences of the applicant.
Consequently, it may make sense to differentiate between the two forms of financing, that is, to derive preferences for the state-financed and student-financed programmes separately.
Nevertheless, applying for both the state-financed and student-financed versions of the same programme has no financial costs, so probably few students employ the strategy presented above.

Following these ideas, the preference matrix of any student can be determined. In order to aggregate them, a weighting scheme should be chosen.
At first sight, it might look that all contributions are equal, so each student should have the same weight, which will be called the \emph{unweighted} problem. In the unweighted preference matrix $A^{UW}$, the entry $a_{ij}^{UW}$ gives the number of applicants who prefer object $X_i$ to object $X_j$.

On the other hand, some students have a longer list of applications, hence, there is more information available on their preferences. Then the unweighted version essentially weights the students according to the number of preferences they have revealed \citep{CaklovicKurdija2017}. The equal contribution of each applicant can be achieved by introducing the weight $w_i = 1/k$ for student $i$ if (s)he has given $k$ preferences after the truncation of objects appearing more than once. This will be called the \emph{weighted} problem. In the weighted preference matrix $A^W$, each applicant, who has revealed at least one preference, increases the sum of entries in the (aggregated) preference matrix by one.

Another solution can be the \emph{moderately weighted} problem when the weight of student $i$ is $w_i = 1/ (\ell-1)$ if (s)he has given a (truncated) preference list of $\ell$ objects. In the moderately weighted preference matrix $A^{MW}$, each applicant, who has revealed at least one preference concerning object $i$, increases the sum of entries in the $i$th row and column of the (aggregated) preference matrix by one.

Finally, with respect to the form of financing, the \emph{adjusted} unweighted $\hat{A}^{UW}$, weighted $\hat{A}^{W}$, and moderately weighted $\hat{A}^{MW}$ preference matrices are introduced, respectively, by obtaining the state-financed and student-financed unweighted, weighted, and moderately weighted preference matrices separately as above, and correspondingly aggregating them.

This procedure allows for the comparison of any types of objects: higher education institutions (universities or colleges), faculties, courses, etc. It is also possible to present customised rankings specific to various types of students.
In the current paper, the ranking of Hungarian faculties will be discussed, as collectively revealed by all applicants participating in the system.

\begin{example} \label{Examp41}
Consider a student with the following list of applications:
\[
\left[
\begin{array}{cccccc}
1 & \text{SE--AOK} & \text{Medicine} & O & N & A \\
2 & \text{PTE--AOK} & \text{Medicine} & O & N & A \\
3 & \text{DE--AOK} & \text{Medicine} & O & N & K \\
4 & \text{SE--AOK} & \text{Medicine} & O & N & A \\
5 & \text{SE--FOK} & \text{Dentistry} & O & N & K \\
\end{array}
\right]
\]
Since the objects are the faculties, SE--AOK appears twice, from which the second is deleted. The preferences of the student over the four faculties are as follows:
\begin{eqnarray*}
\text{SE--AOK} & \succ & \text{PTE--AOK}; \\
\text{SE--AOK} & \succ & \text{DE--AOK}; \\
\text{SE--AOK} & \succ & \text{SE--FOK}; \\
\text{PTE--AOK} & \succ & \text{DE--AOK}; \\
\text{PTE--AOK} & \succ & \text{SE--FOK}; \\
\text{DE--AOK} & \succ & \text{SE--FOK}.
\end{eqnarray*}
Thus the applicant has provided six preferences.

If the state-financed and student-financed forms are treated separately, then there are only two revealed preferences:
\begin{eqnarray*}
\text{SE--AOK} & \succ & \text{PTE--AOK}; \qquad \qquad (A) \\
\text{DE--AOK} & \succ & \text{SE--FOK}. \qquad \qquad \quad (K)
\end{eqnarray*}

The objects are $X_1 = \text{SE--AOK}$, $X_2 = \text{PTE--AOK}$, $X_3 = \text{DE--AOK}$, and $X_4 = \text{SE--FOK}$.
The corresponding unweighted ($A^{UW}$), weighted ($A^{W}$), and moderately weighted ($A^{MW}$), as well as, adjusted unweighted ($\hat{A}^{UW}$), weighted ($\hat{A}^{W}$), and moderately weighted ($\hat{A}^{MW}$) preference matrices are
\[
A^{UW} = \left[
\begin{array}{K{2em} K{2em} K{2em} K{2em}} 
    0     & 1     & 1     & 1 \\
    0     & 0     & 1     & 1 \\
    0     & 0     & 0     & 1 \\
    0     & 0     & 0     & 0 \\
\end{array}
\right] \text{, }
A^{W} = \left[
\begin{array}{K{2em} K{2em} K{2em} K{2em}} 
    0     & 1/6   & 1/6   & 1/6 \\
    0     & 0     & 1/6   & 1/6 \\
    0     & 0     & 0     & 1/6 \\
    0     & 0     & 0     & 0   \\
\end{array}
\right],
\]
\[
A^{MW} = \left[
\begin{array}{K{2em} K{2em} K{2em} K{2em}} 
    0     & 1/3   & 1/3   & 1/3 \\
    0     & 0     & 1/3   & 1/3 \\
    0     & 0     & 0     & 1/3 \\
    0     & 0     & 0     & 0   \\
\end{array}
\right] \text{, and }
\hat{A}^{UW} = \left[
\begin{array}{K{2em} K{2em} K{2em} K{2em}} 
    0     & 1     & 0     & 0   \\
    0     & 0     & 0     & 0   \\
    0     & 0     & 0     & 1   \\
    0     & 0     & 0     & 0   \\
\end{array}
\right],
\]
with $\hat{A}^{UW} = \hat{A}^{W} = \hat{A}^{MW}$, respectively.
\end{example}

As we have already mentioned, once a preference matrix is obtained, any method of paired comparisons-based ranking can be used to rank the objects, including the procedures presented in Section~\ref{Sec31}. In our case, the objects are the faculties, and the row sum gives the number of ``net'' preferences, that is, the difference between the favourable and the unfavourable preferences of the faculty, while the normalised row sum is the ratio of ``net'' preferences to all preferences.
The more complicated least squares method has no such an expressive meaning, but it essentially adjusts ``net'' preferences (row sum) by taking into consideration the prestige of faculties that are compared with the given one. We will see that this modification can have a significant impact on the ranking.

\section{Results} \label{Sec5}

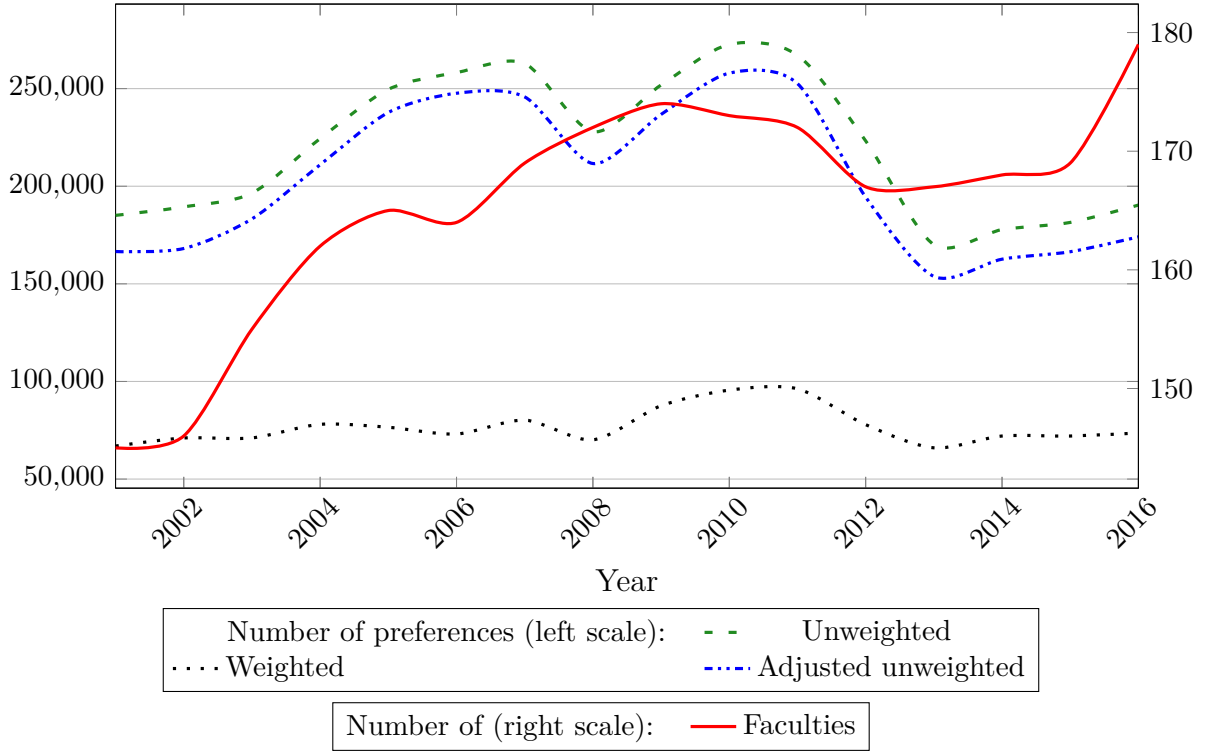
\begin{figure}[ht!]
\centering

\begin{tikzpicture}
\begin{axis}[width = 0.94\textwidth, 
height = 0.5\textwidth,
xmin = 2001,
xmax = 2016,
xlabel = Year,
x tick label style = {/pgf/number format/1000 sep={}},
x tick label style={inner xsep=0pt,rotate=45,anchor=north east},
ymajorgrids,
scaled ticks = false,
y tick label style={/pgf/number format/fixed},
]
\addplot[ForestGreen,smooth,very thick,loosely dashed] coordinates {
(2001,185083)
(2002,189425)
(2003,196464)
(2004,224410)
(2005,249636)
(2006,258197)
(2007,262814)
(2008,227904)
(2009,251860)
(2010,272639)
(2011,267142)
(2012,223221)
(2013,170092)
(2014,177785)
(2015,181459)
(2016,190356)
};
\addplot[black,smooth,very thick,loosely dotted] coordinates {
(2001,66924)
(2002,71066)
(2003,71093)
(2004,78044)
(2005,76534)
(2006,73082)
(2007,80196)
(2008,70112)
(2009,87548)
(2010,95561)
(2011,96225)
(2012,77879)
(2013,66007)
(2014,72052)
(2015,72045)
(2016,73632)
};
\addplot[blue,smooth,very thick,dashdotdotted] coordinates {
(2001,166552)
(2002,168104)
(2003,183283)
(2004,211022)
(2005,237782)
(2006,247631)
(2007,245739)
(2008,211639)
(2009,236796)
(2010,257912)
(2011,252560)
(2012,194439)
(2013,153812)
(2014,162630)
(2015,166451)
(2016,174061)
};
\end{axis}

\begin{axis}[width = 0.94\textwidth, 
height = 0.5\textwidth,
xmin = 2001,
xmax = 2016,
xticklabels={,,},
axis y line* = right,
scaled ticks = false,
y tick label style={/pgf/number format/fixed},
]
\addplot[red,smooth,very thick] coordinates {
(2001,145)
(2002,146)
(2003,155)
(2004,162)
(2005,165)
(2006,164)
(2007,169)
(2008,172)
(2009,174)
(2010,173)
(2011,172)
(2012,167)
(2013,167)
(2014,168)
(2015,169)
(2016,179)
};
\end{axis}
\end{tikzpicture}

\vspace{-0.4cm}
\begin{center}
\begin{tikzpicture}
        \begin{customlegend}[legend columns=2, legend entries={Number of preferences (left scale):$\quad$,Unweighted$\quad$,Weighted \hspace{4cm} $\quad$,Adjusted unweighted}, legend style = {font=\small}]
        \addlegendimage{empty legend}
        \addlegendimage{color=ForestGreen,very thick,loosely dashed}
        \addlegendimage{color=black,very thick,loosely dotted}
        \addlegendimage{color=blue,smooth,very thick,dashdotdotted}
        \end{customlegend}
\end{tikzpicture}

\vspace{0.1cm}

\begin{tikzpicture}        
        \begin{customlegend}[legend columns=4, legend entries={Number of (right scale):$\quad$,Faculties}, legend style = {font=\small}]
        \addlegendimage{empty legend}
		\addlegendimage{color=red,very thick}  
        \end{customlegend}
\end{tikzpicture}
\end{center}

\caption{Descriptive statistics of the dataset, 2001-2016}
\label{Fig2}
\end{figure}


Figure~\ref{Fig2} presents some descriptive statistics of the dataset in the period analysed.\footnote{~In the database, the number of faculties in 2016 was originally $191$ due to the renaming of some higher education institutions. After eliminating the ones without a preference from any applicant, only $179$ remained.}
It can be realised that the Hungarian higher education admission is a huge system with more than 150 thousand revealed preferences in each year. It is under constant reform, even the set of the faculties, the basic units of higher education, changes in almost every year.

\begin{table}[ht!]
  \centering
  \caption{Comparison of the Hungarian Dentistry and Medicine faculties, 2016}
  \label{Table3}
\begin{threeparttable}
\addtocounter{table}{-1}
\begin{subtable}{\textwidth}
  \centering
  \caption{List of Dentistry and Medicine faculties}
  \rowcolors{3}{gray!20}{}
  \label{Table3a}
    \begin{tabularx}{0.8\textwidth}{LLl} \toprule \hiderowcolors
    Faculty & City  & Type \\ \midrule \showrowcolors
    DE--AOK & Debrecen & Medicine \\
    DE--FOK & Debrecen & Dentistry \\
    PTE--AOK & P\'ecs  & Dentistry and Medicine \\
    SE--AOK & Budapest & Medicine \\
    SE--FOK & Budapest & Dentistry \\
    SZTE--AOK & Szeged & Medicine \\
    SZTE--FOK & Szeged & Dentistry \\ \bottomrule
    \end{tabularx}
\end{subtable}

\vspace{0.5cm}
\begin{subtable}{\textwidth}
  \centering
  \caption{The unweighted preference matrix $A^{UW}$ of Dentistry and Medicine faculties}
  \rowcolors{3}{gray!20}{}
  \label{Table3b}   
    \begin{tabularx}{\textwidth}{l CCC CCCC} \toprule \hiderowcolors
    Faculty & F1 & F2 & F3 & F4 & F5 & F6 & F7 \\ \midrule \showrowcolors
    DE--AOK (F1) & 0     & 138   & 506   & 127   & 53    & 308   & 43 \\
    DE--FOK (F2) & 146   & 0     & 144   & 21    & 37    & 52    & 76 \\
    PTE--AOK (F3) & 270   & 87    & 0     & 140   & 84    & 273   & 83 \\
    SE--AOK (F4) & 634   & 72    & 778   & 0     & 244   & 874   & 68 \\
    SE--FOK (F5) & 109   & 178   & 258   & 101   & 0     & 129   & 204 \\
    SZTE--AOK (F6) & 560   & 58    & 835   & 132   & 49    & 0     & 72 \\
    SZTE--FOK (F7) & 45    & 137   & 200   & 17    & 32    & 122   & 0 \\ \bottomrule
    \end{tabularx}
\end{subtable}

\vspace{0.5cm}
\begin{subtable}{\textwidth}
  \centering
  \caption{Scores of Dentistry and Medicine faculties, unweighted preference matrix $A^{UW}$}
  \rowcolors{3}{gray!20}{}
  \label{Table3c}   
    \begin{tabularx}{\textwidth}{l RRR C} \toprule \hiderowcolors
    Faculty & Row sum & Norm. row sum & Least squares & Preferences \\ \midrule \showrowcolors
    DE--AOK & ($-589$) \textbf{6} & ($-0.200$) \textbf{6} & ($-0.176$) \textbf{5} & ($2939$) \textbf{4} \\
    DE--FOK & ($-194$) \textbf{5} & ($-0.169$) \textbf{5} & ($-0.202$) \textbf{6} & ($1146$) \textbf{6} \\
    PTE--AOK & ($-1784$) \textbf{7} & ($-0.488$) \textbf{7} & ($-0.387$) \textbf{7} & ($3658$) \textbf{1} \\
    SE--AOK & ($2132$) \textbf{1} & ($0.665$) \textbf{1} & ($0.531$) \textbf{1} & ($3208$) \textbf{3} \\
    SE--FOK & ($480$) \textbf{2} & ($0.325$) \textbf{2} & ($0.301$) \textbf{2} & ($1478$) \textbf{5} \\
    SZTE--AOK & ($-52$) \textbf{4} & ($-0.015$) \textbf{4} & ($-0.022$) \textbf{3} & ($3464$) \textbf{2} \\
    SZTE--FOK & ($7$) \textbf{3} & ($0.006$) \textbf{3} & ($-0.045$) \textbf{4} & ($1099$) \textbf{7} \\ \bottomrule
    \end{tabularx}
\end{subtable}

\begin{tablenotes}
\item
\vspace{0.25cm}
\footnotesize{Numbers in parentheses indicate the score, bold numbers sign the rank of the faculty}
\end{tablenotes}
\end{threeparttable}
\end{table}

Therefore, in the first step, we restrict our attention to the seven Dentistry and Medicine faculties and to the last recorded year of 2016.\footnote{~Probably these faculties have the most international students in Hungary. For example, their ratio is close to 50\% at SE--AOK and SE--FOK, representing altogether more than $3000$ foreigners with Germany, Iran, Norway, Italy, and South Korea being the top five countries of origin. See at \url{http://semmelweis.hu/english/the-university/facts-and-figures/} (downloaded 31 May 2019).}
The calculations are summarised in Table~\ref{Table3}. Table~\ref{Table3a} shows the main characteristics of these faculties, Table~\ref{Table3b} presents the unweighted preference matrix derived with the methodology described in Section~\ref{Sec3}, and Table~\ref{Table3c} provides the scores and rankings.

Because the unweighted matrix does not take the length of preference lists into account, it is known that $138$ students have preferred DE--AOK to DE--FOK, while $146$ applicants have made the opposite choice. In addition, $2939$ applicants have revealed a preference concerning DE--AOK according to the last column of Table~\ref{Table3c}: there are $1175$ preferences for this faculty (the sum of the first row in Table~\ref{Table3b}) and $1764$ against it (the sum of the first column in Table~\ref{Table3b}), so the number of ``net'' preferences, that is, the row sum is equal to $-589$. Consequently, normalised row sum will be $-589 / 2939 \approx -0.2$.

The ranking obtained from the row sum and normalised row sum methods coincide, while the least squares method changes the position of two pairs of faculties. Since these scores sum up to $0$, one can say that all rural faculties are below the average in this particular set with the exception of the two based at Budapest, the capital of Hungary.

Rankings can be validated not only through their axiomatic properties but by measuring how they reflect the preferences. One way is to count the number of preferences which are contradictory with the ranking:
\begin{equation} \label{eq_pref}
\sum_{X_i, X_j \in N} a_{ij}: j \succ i.
\end{equation}
This value is $2195$ for the row sum and normalised row sum, which increases to $2253$ in the case of the least squares ranking.\footnote{~This can be seen from the preference matrix in Table~\ref{Table3b}, too: DE--FOK is preferred by more applicants to DE--AOK than vice versa, and the same holds in the relation of SZTE--FOK and SZTE--AOK.}
Nonetheless, the seven faculties form only a small sample of the entire dataset, so it is premature to state that the theoretically sound least squares method does not work in practice.

\begin{table}[ht!]
  \centering
  \caption{Kendall rank correlation coefficients, 2016}
  \label{Table4}
    \rowcolors{4}{}{gray!20}
    \begin{tabularx}{\textwidth}{l CCC CCC CCC} \toprule \hiderowcolors
          & \multicolumn{3}{c}{Row sum} & \multicolumn{3}{c}{Norm. row sum} & \multicolumn{3}{c}{Least squares} \\
          & $A^{UW}$ & $A^{W}$ & $\hat{A}^{UW}$ & $A^{UW}$ & $A^{W}$ & $\hat{A}^{UW}$ & $A^{UW}$ & $A^{W}$ & $\hat{A}^{UW}$ \\ \midrule \showrowcolors
    $\mathbf{s} \left( A^{UW} \right)$ & ---   & \emph{0.909} & \emph{0.953} & 0.783 & 0.759 & 0.771 & 0.719 & 0.709 & 0.705 \\
    $\mathbf{s} \left( A^{W} \right)$ &       & ---   & \emph{0.888} & 0.779 & 0.797 & 0.766 & 0.724 & 0.730 & 0.712 \\
    $\mathbf{s} \left( \hat{A}^{UW} \right)$ &       &       & ---   & 0.764 & 0.741 & 0.775 & 0.705 & 0.694 & 0.702 \\
    $\mathbf{p} \left( A^{UW} \right)$ &       &       &       & ---   & \emph{0.902} & \emph{0.944} & 0.831 & 0.826 & 0.820 \\
    $\mathbf{p} \left( A^{W} \right)$ &       &       &       &       & ---   & \emph{0.882} & 0.808 & 0.836 & 0.798 \\
    $\mathbf{p} \left( \hat{A}^{UW} \right)$ &       &       &       &       &       & ---   & 0.815 & 0.815 & 0.830 \\
    $\mathbf{q} \left( A^{UW} \right)$ &       &       &       &       &       &       & ---   & \emph{0.932} & \emph{0.960} \\
    $\mathbf{q} \left( A^{W} \right)$ &       &       &       &       &       &       &       & ---   & \emph{0.927} \\ \hline
    \end{tabularx}
\end{table}

Table~\ref{Table4} shows the (symmetric) Kendall rank correlation coefficients \citep{Kendall1938} between the nine rankings obtained from the unweighted $A^{UW}$, the weighted $A^{UW}$, and the adjusted unweighted $\hat{A}^{UW}$ preference matrices with the three methods presented in Section~\ref{Sec3}. This measure is based on the number of concordant and disconcordant pairs between the two rankings, its value is between $-1$ and $+1$ such that $-1$ indicates complete disagreement, while $+1$ indicates perfect agreement. In order to avoid the adjustment for ties, the number of preferences has been used as a tie-breaking rule to get strict rankings.

It can be seen that the effect of the ranking method is substantially larger than the effect of the preference matrix (compare the italic numbers with the other ones).
Rankings from the unweighted $A^{UW}$ and adjusted unweighted $\hat{A}^{UW}$ preference matrices are more similar than rankings from the weighted version $A^{W}$.
In addition, the least squares method is more robust to the choice of the preference matrix than the other two procedures.
Therefore, in the following analysis we will mainly focus on this procedure and the unweighted preference matrix $A^{UW}$. It is worth noting that the rankings obtained from $A^{UW}$ and from the matrix which considers only the preferences among the state-financed programmes are similar.

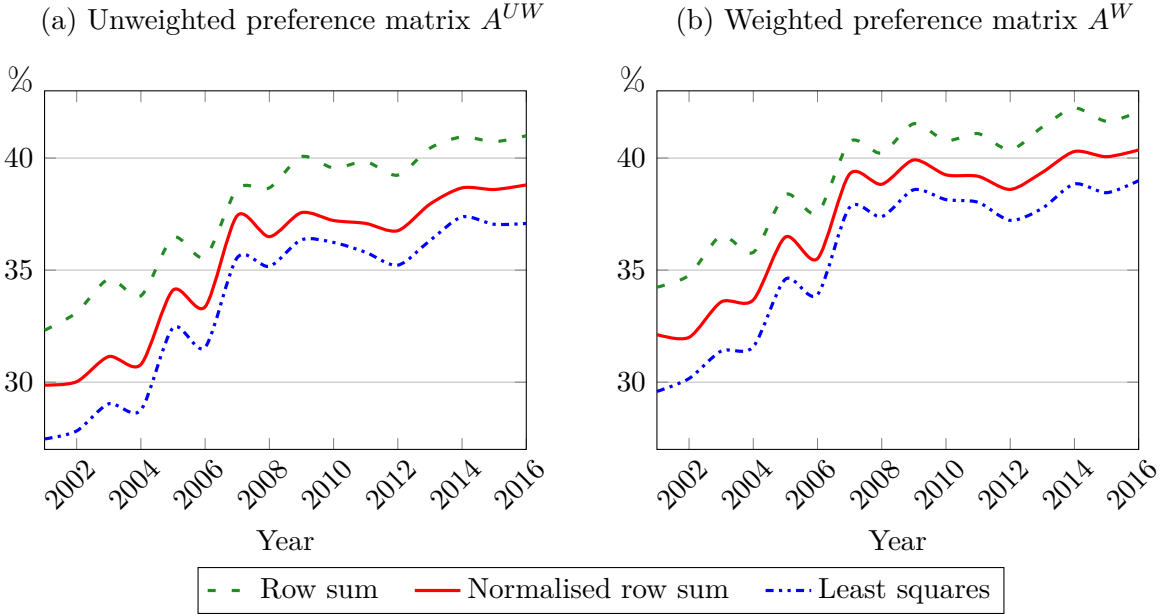
\begin{figure}[ht]
\centering

\begin{subfigure}{0.495\textwidth}
\caption{Unweighted preference matrix $A^{UW}$}
\label{Fig3a}
\begin{tikzpicture}
\begin{axis}[width=\textwidth, 
height=0.8\textwidth,
xmin = 2001,
xmax = 2016,
xlabel = Year,
xlabel style = {font=\small},
ymin = 27,
ymax = 43,
ylabel = \%,
y label style = {at={(axis description cs:-0.05,1.1)},rotate=90,anchor=south},
x tick label style = {/pgf/number format/1000 sep={}},
x tick label style={inner xsep=0pt,rotate=45,anchor=north east},
ymajorgrids,
scaled ticks = false,
y tick label style={/pgf/number format/fixed},
]
\addplot[ForestGreen,smooth,very thick,loosely dashed] coordinates {
(2001,32.3173038468606)
(2002,33.1089108910891)
(2003,34.5971032803338)
(2004,33.8564002888756)
(2005,36.4535887544216)
(2006,35.5885315769082)
(2007,38.6228627315176)
(2008,38.6722694979407)
(2009,40.0626742606582)
(2010,39.557835109037)
(2011,39.8315683836856)
(2012,39.2385033641584)
(2013,40.4406659998823)
(2014,40.942606560055)
(2015,40.7295671021934)
(2016,40.9961222795532)
};
\addplot[red,smooth,very thick] coordinates { 
(2001,29.8703824770508)
(2002,30.0208525801769)
(2003,31.1420921899177)
(2004,30.7909629695646)
(2005,34.1056578378119)
(2006,33.3621227202485)
(2007,37.4329373625454)
(2008,36.4925582701488)
(2009,37.5654632155293)
(2010,37.2122110189665)
(2011,37.0855200604922)
(2012,36.7595342732091)
(2013,37.9459351409825)
(2014,38.6697415417499)
(2015,38.5965975785164)
(2016,38.8020340835067)
};
\addplot[blue,smooth,very thick,dashdotdotted] coordinates {
(2001,27.4687572602562)
(2002,27.8231490035634)
(2003,29.0368718951055)
(2004,28.7603047992514)
(2005,32.4244099408739)
(2006,31.5607849820099)
(2007,35.5582274916861)
(2008,35.1700716090986)
(2009,36.3507504168983)
(2010,36.2350947590037)
(2011,35.7959437303007)
(2012,35.2211485478517)
(2013,36.3197563671425)
(2014,37.3726692353123)
(2015,37.0447318678048)
(2016,37.0810481413772)
};
\end{axis}
\end{tikzpicture}
\end{subfigure}
\begin{subfigure}{0.495\textwidth}
\caption{Weighted preference matrix $A^{W}$}
\label{Fig3b}
\begin{tikzpicture}
\begin{axis}[width = \textwidth, 
height = 0.8\textwidth,
xmin = 2001,
xmax = 2016,
xlabel = Year,
xlabel style = {font=\small},
ymin = 27,
ymax = 43,
ylabel = \%,
y label style = {at={(axis description cs:-0.05,1.1)},rotate=90,anchor=south},
x tick label style = {/pgf/number format/1000 sep={}},
x tick label style={inner xsep=0pt,rotate=45,anchor=north east},
ymajorgrids,
scaled ticks = false,
y tick label style={/pgf/number format/fixed},
]
\addplot[ForestGreen,smooth,very thick,loosely dashed] coordinates {
(2001,34.2411577026962)
(2002,34.7739860867094)
(2003,36.5290133045044)
(2004,35.7946006038107)
(2005,38.3695263978535)
(2006,37.5506443512246)
(2007,40.7220702136906)
(2008,40.2261959096737)
(2009,41.533645760584)
(2010,40.7859004046777)
(2011,41.0908461072137)
(2012,40.3638000096607)
(2013,41.3856611167502)
(2014,42.2316752715628)
(2015,41.6416128808381)
(2016,42.0594759524843)
};
\addplot[red,smooth,very thick] coordinates { 
(2001,32.1249192403038)
(2002,32.0000991700537)
(2003,33.5800390233315)
(2004,33.6680956530015)
(2005,36.4756715658275)
(2006,35.5129528344548)
(2007,39.2864743283716)
(2008,38.8319776792505)
(2009,39.9184719348513)
(2010,39.2521904482849)
(2011,39.1869947667294)
(2012,38.6004479476403)
(2013,39.3622898581461)
(2014,40.2928440570695)
(2015,40.0608878710062)
(2016,40.3604863827318)
};
\addplot[blue,smooth,very thick,dashdotdotted] coordinates {
(2001,29.5879761264377)
(2002,30.16303422841)
(2003,31.3857904435035)
(2004,31.5810922288037)
(2005,34.5941424104071)
(2006,33.9080172174506)
(2007,37.7996171285113)
(2008,37.3875608550128)
(2009,38.5877733466483)
(2010,38.1466238717626)
(2011,38.026069851167)
(2012,37.2198691621127)
(2013,37.7604900490349)
(2014,38.843566683321)
(2015,38.4548083373815)
(2016,38.9946400115889)
};
\end{axis}
\end{tikzpicture}
\end{subfigure}

\vspace{-0.4cm}
\begin{center}
\begin{tikzpicture}
        \begin{customlegend}[legend columns=3, legend entries={Row sum$\quad$,Normalised row sum$\quad$,Least squares}, legend style = {font=\small}]
        \addlegendimage{color=ForestGreen,very thick,loosely dashed}
        \addlegendimage{color=red,very thick}
        \addlegendimage{color=blue,smooth,very thick,dashdotdotted} 
        \end{customlegend}
\end{tikzpicture}
\end{center}

\caption{The ratio of preferences contradictory with the ranking, 2001-2016}
\label{Fig3}

\end{figure}


Figure~\ref{Fig3} illustrates the performance of the three methods by calculating the ratio of preferences which are contradictory with the appropriate ranking according to formula~\eqref{eq_pref}. The normalised row sum procedure turns out to be better than the simple row sum, but the least squares method continuously beats both of them. Thus, the message of Table~\ref{Table2} in favour of the least squares method, which was based on purely theoretical considerations, is reinforced by its superior performance on a large-scale dataset across more than a decade.

The ratio of preferences that are contradictory with the rankings has increased robustly between 2001 and 2016. While it is difficult to disentangle the effect of the constantly changing set of objects, the judgements of the applicants have probably become more diverse in the period considered.

To investigate the dynamics of the results, eight faculties have been chosen for in-depth analysis:
\begin{itemize}
\item
BME--GEK: Faculty of Mechanical Engineering, Budapest University of Technology and Economics (more than 2000 revealed preferences in each year);
\item
BME--GTK: Faculty of Economic and Social Sciences, Budapest University of Technology and Economics (more than 4500 revealed preferences in each year);
\item
PTE--AOK: Medical School, University of P\'ecs (more than 2000 revealed preferences in each year);
\item
SE--AOK: Faculty of Medicine, Semmelweis University (more than 1500 revealed preferences in each year);
\item
ELTE--AJK: Faculty of Law, E\"otv\"os L\'or\'and University (more than 4000 revealed preferences in each year);
\item
ELTE--TTK: Faculty of Science, E\"otv\"os L\'or\'and University (more than 5500 revealed preferences in each year);
\item
SZTE--BTK: Faculty of Humanities and Social Sciences, University of Szeged (more than 3500 revealed preferences in each year); and
\item
ZSKF: King Sigismund University (more than 2000 revealed preferences in each year).
\end{itemize}

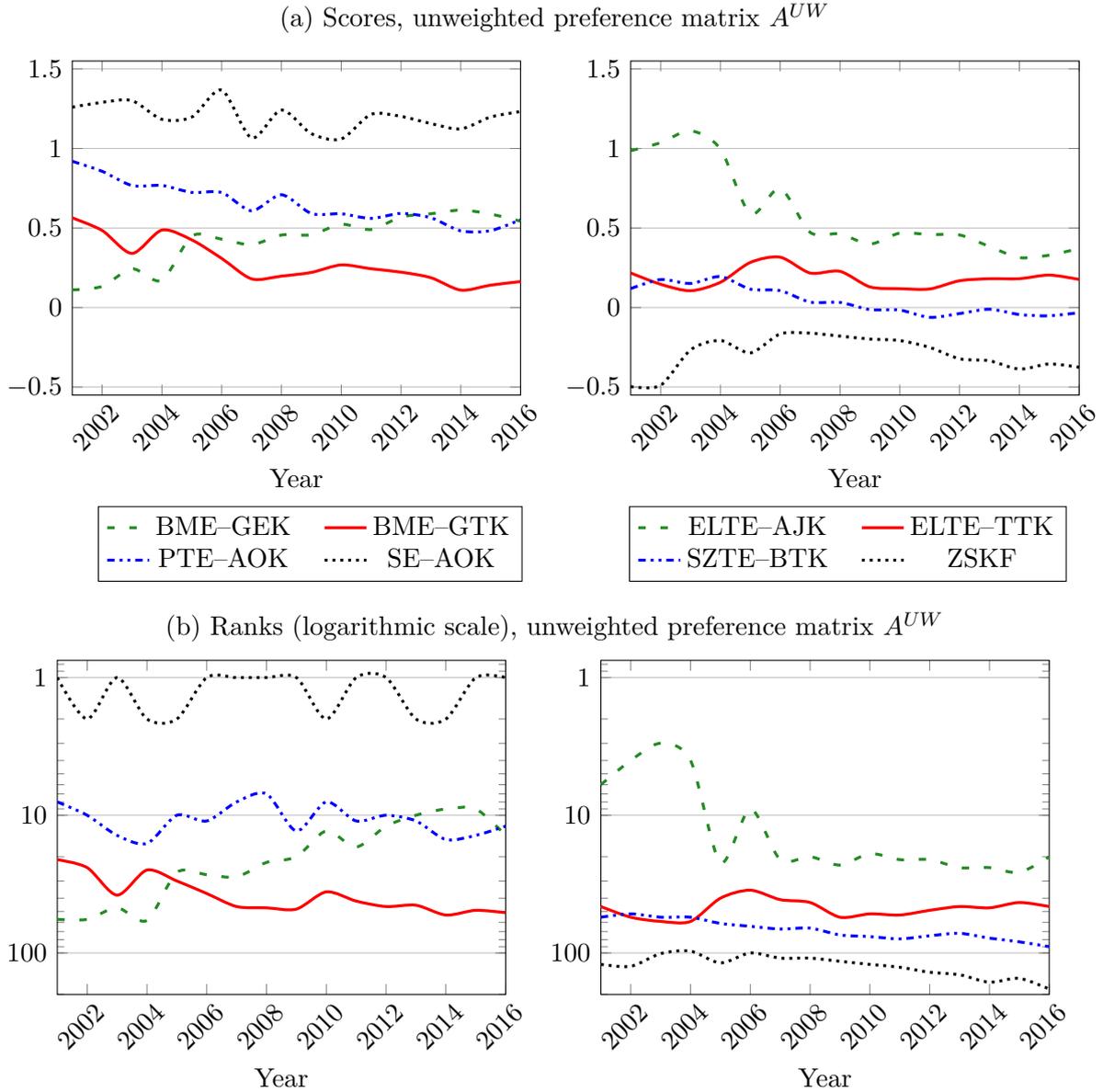
\begin{figure}[ht]
\centering

\begin{subfigure}{\textwidth}
\caption{Scores, unweighted preference matrix $A^{UW}$}
\label{Fig4a}
\begin{tikzpicture}
\begin{axis}[width=0.5\textwidth, 
height=0.4\textwidth,
xmin = 2001,
xmax = 2016,
xlabel = Year,
xlabel style = {font=\small},
ymin = -0.55,
ymax = 1.55,
x tick label style = {/pgf/number format/1000 sep={}},
x tick label style={inner xsep=0pt,rotate=45,anchor=north east},
ymajorgrids,
scaled ticks = false,
y tick label style={/pgf/number format/fixed},
]
\addplot[ForestGreen,smooth,very thick,loosely dashed] coordinates {
(2001,0.11104777989995)
(2002,0.133190700888121)
(2003,0.243853427842068)
(2004,0.171429561094607)
(2005,0.449244525147314)
(2006,0.430145717675869)
(2007,0.395613672359136)
(2008,0.455908819893862)
(2009,0.457887074596244)
(2010,0.523645556257489)
(2011,0.490652102496768)
(2012,0.569473418448652)
(2013,0.589639554195721)
(2014,0.612748225050528)
(2015,0.588117950931223)
(2016,0.542984415696179)
};
\addplot[red,smooth,very thick] coordinates { 
(2001,0.563901885694877)
(2002,0.484795753534196)
(2003,0.340641814171368)
(2004,0.487086623978622)
(2005,0.425433974028426)
(2006,0.309056936016802)
(2007,0.180598108767054)
(2008,0.197609399566495)
(2009,0.22095793105679)
(2010,0.268142163591396)
(2011,0.2438700945156)
(2012,0.222664997040681)
(2013,0.187746580700734)
(2014,0.109829577773497)
(2015,0.141061504591472)
(2016,0.163337486398341)
};
\addplot[blue,smooth,very thick,dashdotdotted] coordinates {
(2001,0.919472067819588)
(2002,0.856551569925337)
(2003,0.767274233035469)
(2004,0.767711908211038)
(2005,0.72387126413776)
(2006,0.722869901271731)
(2007,0.608830737937726)
(2008,0.709273584496993)
(2009,0.591495224653337)
(2010,0.589605568095999)
(2011,0.560544387330922)
(2012,0.59160122415808)
(2013,0.564307506421505)
(2014,0.48231983812324)
(2015,0.483869862783417)
(2016,0.55255950371328)
};
\addplot[black,smooth,very thick,dotted] coordinates {
(2001,1.25949648622004)
(2002,1.29010371833772)
(2003,1.30058491657237)
(2004,1.18305582063557)
(2005,1.19752376959949)
(2006,1.3666576962318)
(2007,1.06639890781563)
(2008,1.24153096369049)
(2009,1.09312919386651)
(2010,1.06007341093856)
(2011,1.21434830364314)
(2012,1.2014588561557)
(2013,1.15666147857728)
(2014,1.1232976268146)
(2015,1.19776446835452)
(2016,1.23183074429447)
};
\end{axis}
\end{tikzpicture}
\begin{tikzpicture}
\begin{axis}[width = 0.5\textwidth, 
height = 0.4\textwidth,
xmin = 2001,
xmax = 2016,
xlabel = Year,
xlabel style = {font=\small},
ymin = -0.55,
ymax = 1.55,
x tick label style = {/pgf/number format/1000 sep={}},
x tick label style={inner xsep=0pt,rotate=45,anchor=north east},
ymajorgrids,
scaled ticks = false,
y tick label style={/pgf/number format/fixed},
]
\addplot[ForestGreen,smooth,very thick,loosely dashed] coordinates {
(2001,0.985967177855641)
(2002,1.03628025720687)
(2003,1.1110821487671)
(2004,0.996358778540705)
(2005,0.589699126064107)
(2006,0.748256078434367)
(2007,0.473098544684534)
(2008,0.463323711178947)
(2009,0.399199506760028)
(2010,0.468228654520045)
(2011,0.458915082375774)
(2012,0.456511080391303)
(2013,0.383708870246416)
(2014,0.314069317018362)
(2015,0.329415511158786)
(2016,0.371247069287409)
};
\addplot[red,smooth,very thick] coordinates { 
(2001,0.217421428007848)
(2002,0.146373293930177)
(2003,0.105874896136782)
(2004,0.158895607702882)
(2005,0.283704910510263)
(2006,0.316619373379033)
(2007,0.21714956175336)
(2008,0.227735703752226)
(2009,0.130146246152391)
(2010,0.118949968233618)
(2011,0.116961572254993)
(2012,0.168536590837769)
(2013,0.181300562959015)
(2014,0.181559289676586)
(2015,0.204271109249992)
(2016,0.176756158532298)
};
\addplot[blue,smooth,very thick,dashdotdotted] coordinates {
(2001,0.118986099286371)
(2002,0.176048494698262)
(2003,0.151456758787721)
(2004,0.19480896265338)
(2005,0.11567511702382)
(2006,0.106434937226059)
(2007,0.0345231884401215)
(2008,0.0324869221629231)
(2009,-0.0128169514462326)
(2010,-0.0156445554319209)
(2011,-0.0616633291996288)
(2012,-0.0377055863149933)
(2013,-0.0106383188215132)
(2014,-0.0435946393637369)
(2015,-0.0518287629099026)
(2016,-0.0319437894164717)
};
\addplot[black,smooth,very thick,dotted] coordinates {
(2001,-0.496604543997275)
(2002,-0.488792488416264)
(2003,-0.266180021661058)
(2004,-0.207014985145919)
(2005,-0.284879994490654)
(2006,-0.166046748498936)
(2007,-0.161028519720868)
(2008,-0.179302529961534)
(2009,-0.198137239334078)
(2010,-0.207288333292203)
(2011,-0.250627671001402)
(2012,-0.321723627554464)
(2013,-0.335015789637275)
(2014,-0.386610677401854)
(2015,-0.354936380573092)
(2016,-0.375442579576782)
};
\end{axis}
\end{tikzpicture}
\end{subfigure}

\vspace{0.1cm}
\hspace{0.75cm}
\begin{tikzpicture}
        \begin{customlegend}[legend columns=2, legend entries={BME--GEK$\quad$,BME--GTK,PTE--AOK$\quad$,SE--AOK}, legend style = {font=\small}]
        \addlegendimage{color=ForestGreen,very thick,loosely dashed}
        \addlegendimage{color=red,very thick}
        \addlegendimage{color=blue,smooth,very thick,dashdotdotted}
        \addlegendimage{color=black,smooth,very thick,dotted} 
        \end{customlegend}
\end{tikzpicture}
\hspace{1.25cm}
\begin{tikzpicture}
        \begin{customlegend}[legend columns=2, legend entries={ELTE--AJK$\quad$,ELTE--TTK,SZTE--BTK$\quad$,ZSKF}, legend style = {font=\small}]
        \addlegendimage{color=ForestGreen,very thick,loosely dashed}
        \addlegendimage{color=red,very thick}
        \addlegendimage{color=blue,smooth,very thick,dashdotdotted}
        \addlegendimage{color=black,smooth,very thick,dotted} 
        \end{customlegend}
\end{tikzpicture}
\vspace{0.4cm}

\begin{subfigure}{\textwidth}
\caption{Ranks (logarithmic scale), unweighted preference matrix $A^{UW}$}
\label{Fig4b}
\begin{tikzpicture}
\begin{axis}[width = 0.5\textwidth, 
height = 0.4\textwidth,
xmin = 2001,
xmax = 2016,
xlabel = Year,
xlabel style = {font=\small},
ymin = 0.75,
ymax = 200,
y dir = reverse,
ymode = log,
x tick label style = {/pgf/number format/1000 sep={}},
x tick label style={inner xsep=0pt,rotate=45,anchor=north east},
ymajorgrids,
scaled ticks = false,
log ticks with fixed point,
y tick label style={/pgf/number format/fixed},
]
\addplot[ForestGreen,smooth,very thick,loosely dashed] coordinates {
(2001,57)
(2002,57)
(2003,47)
(2004,58)
(2005,26)
(2006,27)
(2007,28)
(2008,22)
(2009,20)
(2010,13)
(2011,17)
(2012,12)
(2013,10)
(2014,9)
(2015,9)
(2016,14)
};
\addplot[red,smooth,very thick] coordinates { 
(2001,21)
(2002,24)
(2003,38)
(2004,25)
(2005,30)
(2006,37)
(2007,46)
(2008,47)
(2009,48)
(2010,36)
(2011,42)
(2012,46)
(2013,45)
(2014,53)
(2015,49)
(2016,51)
};
\addplot[blue,smooth,very thick,dashdotdotted] coordinates {
(2001,8)
(2002,10)
(2003,14)
(2004,16)
(2005,10)
(2006,11)
(2007,8)
(2008,7)
(2009,13)
(2010,8)
(2011,11)
(2012,10)
(2013,11)
(2014,15)
(2015,14)
(2016,12)
};
\addplot[black,smooth,very thick,dotted] coordinates {
(2001,1)
(2002,2)
(2003,1)
(2004,2)
(2005,2)
(2006,1)
(2007,1)
(2008,1)
(2009,1)
(2010,2)
(2011,1)
(2012,1)
(2013,2)
(2014,2)
(2015,1)
(2016,1)
};
\end{axis}
\end{tikzpicture}
\begin{tikzpicture}
\begin{axis}[width = 0.5\textwidth, 
height = 0.4\textwidth,
xmin = 2001,
xmax = 2016,
xlabel = Year,
xlabel style = {font=\small},
ymin = 0.75,
ymax = 200,
y dir = reverse,
ymode = log,
x tick label style = {/pgf/number format/1000 sep={}},
x tick label style={inner xsep=0pt,rotate=45,anchor=north east},
ymajorgrids,
scaled ticks = false,
log ticks with fixed point,
y tick label style={/pgf/number format/fixed},
]
\addplot[ForestGreen,smooth,very thick,loosely dashed] coordinates {
(2001,6)
(2002,4)
(2003,3)
(2004,4)
(2005,22)
(2006,9)
(2007,21)
(2008,20)
(2009,23)
(2010,19)
(2011,21)
(2012,21)
(2013,24)
(2014,24)
(2015,26)
(2016,20)
};
\addplot[red,smooth,very thick] coordinates { 
(2001,46)
(2002,55)
(2003,59)
(2004,59)
(2005,40)
(2006,35)
(2007,41)
(2008,43)
(2009,55)
(2010,52)
(2011,53)
(2012,49)
(2013,46)
(2014,47)
(2015,43)
(2016,46)
};
\addplot[blue,smooth,very thick,dashdotdotted] coordinates {
(2001,55)
(2002,52)
(2003,55)
(2004,55)
(2005,61)
(2006,64)
(2007,67)
(2008,66)
(2009,74)
(2010,76)
(2011,79)
(2012,75)
(2013,72)
(2014,78)
(2015,83)
(2016,90)
};
\addplot[black,smooth,very thick,dotted] coordinates {
(2001,121)
(2002,125)
(2003,101)
(2004,97)
(2005,118)
(2006,100)
(2007,109)
(2008,109)
(2009,115)
(2010,121)
(2011,127)
(2012,138)
(2013,144)
(2014,163)
(2015,153)
(2016,183)
};
\end{axis}
\end{tikzpicture}
\end{subfigure}

\caption{Least squares scores and ranks of selected Hungarian faculties, 2001-2016}
\label{Fig4}

\end{figure}


Figure~\ref{Fig4} shows the least squares scores (Figure~\ref{Fig4a}) and ranks (Figure~\ref{Fig4b}) of these faculties. SE--AOK has achieved at least the second position in each year between 2001 and 2016. PTE--AOK has been close to the bottom of the top 10 consistently.
On the other hand, BME--GTK, as well as ELTE--AJK (and, to some extent, SZTE-BTK) have displayed a declining performance, especially law studies have become less popular among the applicants in this one and a half decade. The gradual improvement of BME--GEK demonstrates that an appropriate long-term strategy can yield significant gains. ELTE--TTK has remained a strong middle-rank faculty with some fluctuations, while ZSKF has not managed to increase its low prestige.

Tables~\ref{Table_A1}, \ref{Table_A2}, and \ref{Table_A3} in the Appendix present the scores and ranks of the faculties in 2016, as obtained from the unweighted $A^{UW}$, the weighted $A^{UW}$, and the adjusted unweighted $\hat{A}^{UW}$ preference matrices, respectively.
It is not surprising that some small faculties have a good position according to the normalised row sum and least squares methods. Turning to the more popular institutions and concentrating on the suggested least squares method, all of the seven Dentistry and Medicine faculties are among the top faculties in Hungary.
While the worse faculties of this set (such as DE--AOK, DE--FOK, or PTE--AOK, see Table~\ref{Table3}) do not seem excellent by the two local measures of row sum and normalised row sum, their performance significantly improves after taking the whole structure of the network into account because, although they are not favoured by the applicants over the leading Dentistry and Medicine faculties, they are still preferred to faculties in other subject areas. Similar causes are behind the relatively better performance of the four Pharmacy faculties (DE--GYTK in Debrecen, PTE--GYTK in P\'ecs, SE--GYTK in Budapest, SZTE--GYTK in Szeged) when the least squares method is applied.

On the other hand, the most prestigious business (BCE--GTK) and economics faculties (BCE--KTK) are outside the top 20 according to the least squares ranking, despite that the former is among the best according to row sum, and performs better even with the normalised row sum method.
The unique Veterinary Medicine faculty (SZIE--AOTK) and the leading Faculty of Architecture (BME--ESZK) are also in the top 10, however, their good positions are still revealed by the normalised row sum method. There is no substantial difference between the normalised row sum and least squares rankings in the case of BCE--TK (Faculty of Social Sciences and International Relations at Corvinus University of Budapest), too. 
Nonetheless, certain faculties gain (e.g. PPKE--ITK, Faculty of Information Technology and Bionics at P\'azm\'any P\'eter Catholic University), or lose (e.g. ELTE--GYFK, B\'arczi Guszt\'av Faculty of Special Needs Educations at E\"otv\"os Lor\'and University; NKE--RTK, Faculty of Law Enforcement at National University of Public Service; TE, University of Physical Education) from the use of the least squares method.

\begin{table}[ht!]
  \centering
  \caption{Ranks of selected Hungarian faculties, 2016}
  \label{Table5}
    \rowcolors{4}{}{gray!20}
    \begin{tabularx}{\textwidth}{l CCC CCC CCC} \toprule \hiderowcolors
    Method $\rightarrow$     & \multicolumn{3}{c}{Row sum} & \multicolumn{3}{c}{Norm. row sum} & \multicolumn{3}{c}{Least squares} \\
    Faculty $\downarrow$     & $A^{UW}$ & $A^{W}$ & $\hat{A}^{UW}$ & $A^{UW}$ & $A^{W}$ & $\hat{A}^{UW}$ & $A^{UW}$ & $A^{W}$ & $\hat{A}^{UW}$ \\ \midrule \showrowcolors
    BCE--GTK & 2     & 1     & 2     & 12    & 12    & 12    & 21    & 21    & 22 \\
    BCE--KTK & 30    & 34    & 30    & 28    & 30    & 22    & 22    & 25    & 19 \\
    BCE--TK & 12    & 13    & 13    & 25    & 27    & 27    & 25    & 32    & 25 \\
    BME--ESZK & 22    & 23    & 22    & 6     & 6     & 7     & 10    & 7     & 9 \\
    BME--TTK & 48    & 115   & 53    & 52    & 112   & 53    & 38    & 51    & 39 \\
    DE--AOK & 19    & 27    & 19    & 38    & 35    & 39    & 9     & 13    & 8 \\
    DE--FOK & 40    & 49    & 39    & 42    & 47    & 42    & 8     & 8     & 7 \\
    DE--GYTK & 149   & 110   & 153   & 162   & 143   & 165   & 50    & 50    & 51 \\
    ELTE--GYFK & 13    & 8     & 10    & 15    & 13    & 15    & 31    & 27    & 30 \\
    NKE--RTK & 20    & 12    & 21    & 24    & 18    & 24    & 39    & 28    & 41 \\
    PPKE--ITK & 51    & 40    & 50    & 53    & 46    & 52    & 33    & 29    & 32 \\
    PTE--AOK & 169   & 144   & 173   & 99    & 109   & 105   & 12    & 14    & 12 \\
    PTE--GYTK & 144   & 128   & 148   & 175   & 173   & 174   & 62    & 94    & 64 \\
    SE--AOK & 1     & 3     & 1     & 2     & 1     & 2     & 1     & 1     & 1 \\
    SE--FOK & 10    & 24    & 8     & 8     & 7     & 6     & 2     & 2     & 2 \\
    SE--GYTK & 122   & 48    & 135   & 87    & 53    & 91    & 16    & 15    & 18 \\
    SZIE--AOTK & 16    & 19    & 15    & 3     & 4     & 3     & 4     & 4     & 5 \\
    SZTE--AOK & 7     & 26    & 7     & 22    & 29    & 21    & 3     & 6     & 4 \\
    SZTE--FOK & 32    & 39    & 31    & 32    & 28    & 32    & 5     & 3     & 6 \\
    SZTE--GYTK & 126   & 94    & 130   & 133   & 113   & 134   & 36    & 37    & 35 \\
    TE    & 11    & 5     & 9     & 17    & 16    & 16    & 37    & 36    & 38 \\ \bottomrule
    \end{tabularx}
\end{table}

Table~\ref{Table5} summarises the ranks of the faculties mentioned above, classified by the three methods and the three variants of preference matrices.
It also reveals that the unweighted $A^{UW}$ and adjusted unweighted $\hat{A}^{UW}$ preference matrices lead to almost the same ranking, that is, separation of preferences with respect to the financing of the tuition has only marginal effects.
On the other hand, there are some remarkable differences between the rankings obtained from the unweighted $A^{UW}$ and weighted $A^{W}$ preference matrices, especially for the row sum and normalised row sum methods.
The case of GYTKs is probably explained by the fact that Pharmacy faculties are ``substitutes'' of Medicine faculties for several applicants (but not vice versa), so the preference lists of students with an unfavourable view on pharmacy are inherently longer.
The issue of BME--TTK (Faculty of Natural Sciences, Budapest University of Technology and Economics) remains to be explored.

Thus the least squares ranking of the faculties has an obvious, intuitive explanation. Concisely, the dataset reveals that the number of applicants who want to be a doctor but choose another field if this dream is not achievable is significantly greater than the number of applicants employing an opposite strategy.
Naturally, one can eliminate this effect by composing separate lists on different subject areas, for example, by considering only an appropriate submatrix of the whole preference matrix as in Table~\ref{Table3}.

However, sometimes there is a demand for a universal ranking. This is what we have provided here.

\section{Conclusions} \label{Sec6}

In this paper, a university ranking has been constructed from the lists of applications, which can be implemented in any system using centralised admissions. It is clear that the proposed ranking has a different nature with respect to other university rankings, and our approach has its own limitations since the preferences are not observed directly but derived from the applications, the preferences of the students are determined not only by the quality of the institutions, and past reputation can increase the inertia of the ranking.
On the other hand, the dataset reflects the collective information of thousands of applicants at a moment of high-stakes decisions, and this collective wisdom of the crowd can perhaps be competitive with some composite indices devised by a dozen academics. To summarise, preference-based rankings do not solve the problems of other methodologies but can be used as an alternative, for example, to check the robustness of traditional university rankings.

We have presented a case study by ranking all faculties in the Hungarian higher education between 2001 and 2016. Three different methods and three variants of preference matrices have been considered for this purpose.
Our results show that the suggested ranking possesses favourable theoretical properties and performs well in practice: the least squares method is hardly sensitive to the aggregation of individual preferences, and it reflects the revealed preferences better than the other procedures discussed.

\section*{Acknowledgements}
\addcontentsline{toc}{section}{Acknowledgements}
\noindent
We would like to thank \emph{P\'eter Bir\'o}, \emph{S\'andor Boz\'oki}, \emph{Tam\'as Halm}, \emph{L\'aszl\'o \'A. K\'oczy}, \emph{Lajos R\'onyai}, and \emph{Andr\'as Telcs} for useful advice. \\
Four anonymous reviewers gave valuable comments and suggestions on an earlier draft. \\
The data were provided by the Databank of the Centre for Economic and Regional Studies. \\
The research was financed by the MTA Premium Postdoctoral Research Program grant PPD2019-9/2019. \\
\emph{Csaba T\'oth} acknowledges the support of the Hungarian Academy of Sciences under its Cooperation of Excellences Grant (KEP-6/2018).

\bibliographystyle{apalike}
\bibliography{All_references}

\clearpage

\section*{Appendix}
\addcontentsline{toc}{section}{Appendix}

\renewcommand\thetable{A.\arabic{table}}
\setcounter{table}{0}

\rowcolors{3}{gray!20}{}
\begin{small}
\begin{longtable}{l RRR R}
\caption[Rankings of Hungarian faculties, unweighted preference matrix $A^{UW}$, 2016]{\normalsize{Rankings of Hungarian faculties, unweighted preference matrix $A^{UW}$, 2016} \\ \vspace{0.25cm}
\footnotesize{Numbers in parentheses indicate the score, bold numbers sign the rank of the faculty}}
\label{Table_A1} \\

\hiderowcolors \toprule
Faculty & Row sum & Norm. row sum & Least squares & Preferences \\ \hline \showrowcolors
\endfirsthead

\hiderowcolors \toprule 
\multicolumn{5}{c}%
{{\tablename\ \thetable{} -- continued from the previous page}} \\ \midrule
Faculty & Row sum & Norm. row sum & Least squares & Preferences \\ \hline \showrowcolors
\endhead

\hiderowcolors \hline \multicolumn{5}{r}{{Continued on the next page}} \\ \bottomrule \showrowcolors
\endfoot

\hline
\endlastfoot
    ANNYE & ($21$) \textbf{43} & ($0.6$) \textbf{20} & ($0.372$) \textbf{19} & ($105$) \textbf{157} \\
    AVKF  & ($-383$) \textbf{159} & ($0.366$) \textbf{166} & ($-0.313$) \textbf{154} & ($1431$) \textbf{93} \\
    BCE--ETK & ($-22$) \textbf{78} & ($0.359$) \textbf{171} & ($-0.216$) \textbf{137} & ($78$) \textbf{161} \\
    BCE--GTK & ($2913$) \textbf{2} & ($0.638$) \textbf{12} & ($0.37$) \textbf{21} & ($10547$) \textbf{1} \\
    BCE--KERTK & ($12$) \textbf{45} & ($0.579$) \textbf{26} & ($0.068$) \textbf{61} & ($76$) \textbf{162} \\
    BCE--KTK & ($207$) \textbf{30} & ($0.578$) \textbf{28} & ($0.364$) \textbf{22} & ($1327$) \textbf{96} \\
    BCE--TAJK & ($1$) \textbf{51} & ($0.516$) \textbf{47} & ($0.131$) \textbf{53} & ($31$) \textbf{165} \\
    BCE--TK & ($821$) \textbf{12} & ($0.583$) \textbf{25} & ($0.32$) \textbf{25} & ($4951$) \textbf{17} \\
\end{longtable}
\end{small}

\clearpage
\rowcolors{3}{gray!20}{}
\begin{small}
\begin{longtable}{l RRR R}
\caption[Rankings of Hungarian faculties, weighted preference matrix $A^{W}$, 2016]{\normalsize{Rankings of Hungarian faculties, weighted preference matrix $A^{W}$, 2016} \\ \vspace{0.25cm}
\footnotesize{Numbers in parentheses indicate the score, bold numbers sign the rank of the faculty}}
\label{Table_A2} \\

\hiderowcolors \toprule
Faculty & Row sum & Norm. row sum & Least squares & Preferences \\ \hline \showrowcolors
\endfirsthead

\hiderowcolors \toprule 
\multicolumn{5}{c}%
{{\tablename\ \thetable{} -- continued from the previous page}} \\ \midrule
Faculty & Row sum & Norm. row sum & Least squares & Preferences \\ \hline \showrowcolors
\endhead

\hiderowcolors \hline \multicolumn{5}{r}{{Continued on the next page}} \\ \bottomrule \showrowcolors
\endfoot

\hline
\endlastfoot
    ANNYE & ($12.6$) \textbf{47} & ($0.624$) \textbf{17} & ($0.38$) \textbf{18} & ($50.8$) \textbf{157} \\
    AVKF  & ($-135.27$) \textbf{157} & ($0.384$) \textbf{162} & ($-0.258$) \textbf{154} & ($584.8$) \textbf{88} \\
    BCE--ETK & ($-13$) \textbf{87} & ($0.368$) \textbf{166} & ($-0.184$) \textbf{136} & ($49.4$) \textbf{158} \\
    BCE--GTK & ($1274.77$) \textbf{1} & ($0.644$) \textbf{12} & ($0.357$) \textbf{21} & ($4414.5$) \textbf{1} \\
    BCE--KERTK & ($3.27$) \textbf{54} & ($0.534$) \textbf{41} & ($0.002$) \textbf{73} & ($47.4$) \textbf{160} \\
    BCE--KTK & ($66.63$) \textbf{34} & ($0.559$) \textbf{30} & ($0.324$) \textbf{25} & ($562.63$) \textbf{91} \\
    BCE--TAJK & ($0.63$) \textbf{59} & ($0.517$) \textbf{48} & ($0.176$) \textbf{43} & ($18.23$) \textbf{164} \\
    BCE--TK & ($253.83$) \textbf{13} & ($0.573$) \textbf{27} & ($0.271$) \textbf{32} & ($1746.57$) \textbf{21} \\
\end{longtable}
\end{small}

\clearpage
\rowcolors{3}{gray!20}{}
\begin{small}
\begin{longtable}{l RRR R}
\caption[Rankings of Hungarian faculties, adjusted unweighted preference matrix $\hat{A}^{UW}$, 2016]{\normalsize{Rankings of Hungarian faculties, adjusted unweighted preference matrix $\hat{A}^{UW}$, 2016} \\ \vspace{0.25cm}
\footnotesize{Numbers in parentheses indicate the score, bold numbers sign the rank of the faculty}}
\label{Table_A3} \\

\hiderowcolors \toprule
Faculty & Row sum & Norm. row sum & Least squares & Preferences \\ \hline \showrowcolors
\endfirsthead

\hiderowcolors \toprule 
\multicolumn{5}{c}%
{{\tablename\ \thetable{} -- continued from the previous page}} \\ \midrule
Faculty & Row sum & Norm. row sum & Least squares & Preferences \\ \hline \showrowcolors
\endhead

\hiderowcolors \hline \multicolumn{5}{r}{{Continued on the next page}} \\ \bottomrule \showrowcolors
\endfoot

\hline
\endlastfoot
    ANNYE & ($26$) \textbf{41} & ($0.638$) \textbf{13} & ($0.441$) \textbf{16} & ($94$) \textbf{157} \\
    AVKF  & ($-380$) \textbf{163} & ($0.362$) \textbf{168} & ($-0.33$) \textbf{161} & ($1376$) \textbf{92} \\
    BCE--ETK & ($-22$) \textbf{75} & ($0.359$) \textbf{170} & ($-0.219$) \textbf{138} & ($78$) \textbf{159} \\
    BCE--GTK & ($2606$) \textbf{2} & ($0.64$) \textbf{12} & ($0.385$) \textbf{22} & ($9296$) \textbf{2} \\
    BCE--KERTK & ($12$) \textbf{43} & ($0.579$) \textbf{25} & ($0.073$) \textbf{61} & ($76$) \textbf{161} \\
    BCE--KTK & ($198$) \textbf{30} & ($0.589$) \textbf{22} & ($0.4$) \textbf{19} & ($1114$) \textbf{104} \\
    BCE--TAJK & ($4$) \textbf{46} & ($0.571$) \textbf{29} & ($0.278$) \textbf{31} & ($28$) \textbf{164} \\
    BCE--TK & ($688$) \textbf{13} & ($0.577$) \textbf{27} & ($0.315$) \textbf{25} & ($4476$) \textbf{21} \\
\end{longtable}
\end{small}

\end{document}